\documentclass[final,5p,times,twocolumn,dvipsnames]{elsarticle}

\usepackage{amsmath,amssymb}  
\usepackage{graphicx}         
\usepackage{hyperref}         
\usepackage{xcolor,here}
\usepackage{lineno}           
\usepackage{multicol, algorithm}         
\usepackage{algpseudocode}
\usepackage{listings}
\usepackage{tikz}
\usetikzlibrary{shapes,arrows,calc}

\newtheorem{prop}{\sc Proposition}
\newtheorem{rem}{\sc Remark}
\newtheorem{lem}{\sc Lemma}

\newcommand{\e}{\vskip 0.5mm \ \\}

\journal{Automatica}

\begin{document}

\begin{frontmatter}

\title{A Nonlinear Model Predictive Control  Perspective on Gradient-Based Optimization:\\  A New Efficient, Parameter-Free and Provably Stable Algorithm}

\author[gipsa]{Mazen Alamir\corref{cor1}}

\cortext[cor1]{Corresponding author}
\ead{mazen.alamir@grenoble-inp.fr}
\address[gipsa]{Univ. Grenoble Alpes, CNRS, Grenoble INP, GIPSA-lab, 38000 Grenoble, France}

\begin{abstract}
    This paper discusses some aspects related to gradient-based optimization algorithms with special focus on the requirements associated to their use in the implementation of Nonlinear Model Predictive Control. Based on a dedicated discussion, a new algorithm, termed Search and Accelerate (\textbf{SaA}) is proposed that mixes together a novel line search, a trust region mechanism together with an adaptation of the gradient acceleration scheme. A dedicated benchmark involving a set of 600 instances of box constrained optimization problems is designed and used in order to show the algorithm performances which make it a highly competitive general purpose gradient-based alternative for box-constrained optimization problems. An appealing feature of the algorithm is its robustness to the choice of the few parameters involved in its definition making the default values a valid option for any problem without a priori knowledge of the related Lipchitz constant. Moreover, an example of use of the proposed algorithm in NMPC implementation is proposed showing the possibility to reduce the control updating period which might be mandatory in some circumstances. 
\end{abstract}

\begin{keyword}
    Gradient-Based Optimization \sep Line Search \sep Box-Constraints \sep Gradient Acceleration \sep  Trust-Region\sep NMPC.
\end{keyword}

\end{frontmatter}

\section{Introduction}
\noindent There are hundreds, if not thousands, of variants of gradient-based algorithms for box-constrained optimization. Performing a detailed comparison among them would go beyond any single decently short contribution. This introduction aims at recalling the major alternatives, highlighting the choices they incorporate regarding the different paradigms involved in a descent method. Moreover, the impact of these choices on the use of the algorithm in repetitive solving of dynamically varying, difficult to characterize Nonlinear Model Predictive Control (\textbf{NMPC}) problems is analyzed. By so doing, the objective of this introduction is to provide a clear positioning of the proposed algorithm that underlines its novelties and convenience as a general purpose gradient-based solver as well as an eligible and viable option for use in NMPC implementation. To this end, the competing alternatives that are used in the comparison are smoothly introduced and discussed.
\e Before, let us recall why and when to use first order gradient-based algorithms in NMPC implementation in the first place. 
\subsection{Why / When to use first order methods in NMPC?}
\noindent Nonlinear Model Predictive Control (\textbf{NMPC}) is the most advanced and generic  control design methodology for constrained nonlinear systems \cite{RawlingsMayne2009}. Its popularity is due to its ability to handle constraints, nonlinearities and control problem's formulations that go beyond regulation and tracking. The counterpart to these advantages is the need to repeatedly solve complex optimization problems on-line. More precisely, at each updating period, a new optimization problem is formulated based on the current values of the estimated value of the state, the set-point and possibly the current estimate of the system's parameters and the problem is solved within the available time, namely the control updating period. This amount of computation time depends on the level of uncertainties/disturbance that needs to be compensated for by the virtue of the control feedback.
\e This task can \textit{theoretically} be achieved using state-of-the-art, generally second order methods which are made accessible through user-friendly frameworks \cite{Andersson2019CasADi, Gillis2020, Verschueren2021acados}. Nevertheless, many contributions \cite{PeyrlHelfried2014Piot, RichterS.2012CCCf, AlamirMazen2017ASUP, ALAMIR201565, patrinos2014, NECOARA201549} underlined that, when systems with fast dynamics and short control updating period are considered, the \textit{elementary} computation cost of a single iteration of the optimization algorithm being used might become a critical parameter as it determines the upper bound for the feedback bandwidth \footnote{The process of executing only a small number of iterations and to apply the resulting current solution that might be far from the optimal one is sometimes termed as \textit{distributing the iterations over the system's life-time}  or more shortly \textit{real-time iterations}}. To say it differently: \textit{For sufficiently fast dynamical systems, the available computation time may be so limited that performing only few iterations of a cheap gradient-based optimization algorithm becomes the only eligible strategy for implementing} NMPC. 
\e Historically, the awareness of the previously stated fact came together with an increasing popularity of the Nesterov acceleration scheme (also called the Fast Gradient Method \cite{nesterovfg2004}) and hence many attempts emerged in which NMPC implementation were based on this approach that is briefly recalled hereafter. 
\subsection{The Fast Gradient Method (FGM)}
\noindent Among all first order methods, the Fast Gradient Method \cite{nesterovfg2004} \text{(\textbf{FGM})} has been the most widely used within the NMPC community. Indeed, it can be shown \cite{Alamir2026MPC} that under some circumstances and provided that the parameters of the \textbf{FGM} are appropriately tuned, it is possible to decrease the computation time making the NMPC real-time compatible without noticeably deteriorating the closed-loop behavior. 
\e Notice that in the \textbf{FGM} method, the iterate $y$ that should converge to a solution of the optimization problem is updated according to:
\begin{subequations}\label{FGeq}
\begin{align}
v_{k+1}&= y_k - \alpha g(y_k)\hfill &\text{\color{Gray}(Gradient step)}\label{FG1}\\
y_{k+1}&=v_{k+1} + \beta_k(v_{k+1} - x_k)  \hfill &\text{\color{Gray} (Acceleration step)}\label{FG2}
\end{align}    
\end{subequations}
where the stability constraint on the step size $\alpha$ involved in \eqref{FG1} is expressed by $\alpha\le 1/L$, $L$ being the Lipchitz constant of the cost function while $g$ represent the gradient. The definition of the dynamic gain $\beta_k$ lies in the heart of the \textbf{FGM}. 
\e 
Notice however that, in the context of NMPC, it is very difficult to \textit{appropriately tune} the parameters of \textbf{FGM} so that a decrease of the cost function is achieved over all possible instances of the NMPC-related problem. This is because the Lipchitz constant $L$ invoked above depends on the current state, the amount of change in the set-point asked by the operator and the stage cost-related penalties being used that might be changed by the designer based on the examination of the closed-loop behavior. Consequently, it seems obvious that enforcing an unconditional stability of \textbf{FGM} \cite{nesterovfg2004} through the condition $\alpha\le 1/L$ without adopting too small values of $\alpha$ might be risky. 
\e Now taking $\alpha$ excessively small so that the theoretical condition is satisfied over all possible real-life instances of the optimization problem would very likely lead to slow convergence that would make the closed-loop unsuccessful.  On the other hand, taking high values of the step size $\alpha$ might accelerate the convergence for some instances of the problem at the price of risking divergence in other instances which is obviously unacceptable for critical engineered systems. Such behaviors materialize in the numerical study provided in the paper as far as the implementation of \textbf{FG} is concerned (see Section \ref{sec-bench}).
\e Another reason for which the \textbf{FGM} is not convenient in the  NMPC context is the iteration-dependent scheduled gain $\beta_k$ used in \eqref{FG2} which suggests that in order to cope with the theory, one needs to re-initialize  $v_0,y_0,\beta_0$ each time the horizon is shifted which might be incompatible with the use of warm start.
\e The previous discussion suggests that \textbf{FGM}, which represented a key advance compared to the trivial pure gradient method, is probably \textit{too simple} to be used in realistic and viable industrial NMPC implementation. It is rather limited to static problems for which iterative tuning can be done until success is achieved. 
\subsection{Strong Armijo-Wolf method}
\noindent The sensitivity of the fast gradient is mainly due to the absence of \textit{line-search} to determine the \textit{current} appropriate gradient step. Such a line search is implemented in the famous Strong Armijo-Wolf (\textbf{SAW}) algorithm \cite{MoreThuente1994, HagerZhang2005} and its many variants. This algorithm uses backtracking (successive reduction steps on $\alpha$) in order to enforce both a decrease of the cost function (Armijo rule) and a decrease in the norm of the directional gradient (curvature condition). This is obtained at the price of evaluating the gradient at different points through the back-tracking iterative process within a single iteration. This might increase the elementary cost of a single iteration discussed above. 
\e As a consequence of this adaptation mechanism, the \textbf{SAW} algorithm does not incorporate any acceleration step as the theoretical stability under the acceleration step is based on the use of a constant step size $\alpha\le 1/L$. In other words, the benefit from optimizing the step $\alpha$ by line-search might be negatively over-compensated by the absence of the acceleration step.
\subsection{The FISTA method}
\noindent In order to be able to use stable acceleration step in the absence of a priori knowledge of $L$ while avoiding taking a \textit{universally small value of $\alpha$}, the Fast Iterative Shrinkage-Thresholding Algorithm (\textbf{FISTA}) \cite{BeckTeboulle2009, KimFessler2018} does not attempt a line-search for $\alpha$. Rather it looks for an estimation $\hat L$ of the Lipchitz constant $L$ via a backtracking mechanism\footnote{Successive reduction of the value of $\alpha$ initially taken sufficiently high.} and this enables to use the original acceleration scheme with a scheduled momentum gain. Notice however that by dropping the line search and setting the step size to $\alpha=1/\hat L$, the choice might still be conservative. More precisely, the estimation process needs to start from some initial guess  $L_0$ and the choice of this key quantity is critical as the estimation of $L$ adopts one-side adaptation ($L$ can only increase). Consequently, if very high $L_0$ is used, slow convergence is induced while the use of too small a value leads to slow adaptation process.
\e This possible lack of performance, while already debatable in general, may become a strong disadvantage in the context of fast NMPC for the reasons discussed earlier. 
\subsection{The Adam's method}
\noindent It is important to underline that the comparative study does not include the nowadays very famous \textbf{Adam} stochastic optimization algorithm \cite{KingmaBa2015} that is widely used to train Deep Neural Networks. The reason is that this algorithm is obviously oriented towards off-line long training processes where the step is non uniformly adapted over the different components of the vector of decision variables and this based on a high number of previous computation instances of the gradient over the process. This is obviously incompatible with the real-time computation of continuously varying optimization problems and the need to provide a sufficiently good solution in few iterations in the absence of previous high number of evaluations of the gradient. 
\e Based on the above discussion, the contribution of the present paper lies in the following items: 
\begin{itemize}
\item A new algorithm is proposed that simultaneously conciliates line search, acceleration and stability altogether. The relevance of this contribution is not necessarily linked to the use of the algorithm in NMPC context and can be viewed as a general contribution to the literature on the gradient based iterations for its own sake.
\item A benchmark of 600 optimization problems is built and made available \cite{mazen_alamir_2026_kaggle_bench_opt} in order to evaluate the performance of the proposed algorithm and to draw comparison with the previously recalled alternatives to solve box-constraints optimization problem regardless of its use in the NMPC context.
\item Finally, numerical experiments regarding the use of the algorithm in NMPC implementation are conducted in order to show the relevance of its use when very small updating periods are necessary that go below the minimal elementary computation times corresponding to higher order methods. While this fact is already reported for the \textbf{FGM} (when successful), these experiments show that this feature still hold for the proposed algorithm that incorporates slightly more involved set of computation than the \textbf{FGM}.
\end{itemize}
 The paper is organized as follows. First Section \ref{sec-defnotation} introduces some definitions and notation used throughout the paper. Section \ref{sec-algo} presents the proposed algorithm and briefly analyzes its properties. Section \ref{sec-bench} describes the benchmark used in the numerical investigation and comparison that are presented in Section \ref{sec-results}. This section also shows the NMPC-related investigation. The paper ends with a conclusion summarizing its findings and describing some further investigation to come.
\section{Definitions and notation}\label{sec-defnotation}
\subsection{The optimization problem}
\noindent We are interested in the following box-constrained optimization problem:
\begin{equation}
\min_{x\in [\underline x, \bar x]}f(x)\label{optproblem}
\end{equation}
where $x\in \mathbb R^n$ is the decision variable while $f$ is a differentiable function with gradient denoted by $g(x):=\nabla f(x)\in \mathbb R^n$. 
\e Notice that the fact that \eqref{optproblem} allows only for box constraints to be considered means that if constraints of the form: 
\begin{equation}
c_j(x)\le 0 \quad j\in \{1,\dots,n_c\}\label{defdesconstraintes}
\end{equation}
need to be handled, they should be added to some original cost function, say $f_0$, through strongly weighted exact penalties:
\begin{equation}
f(x) := f_0(x)+\rho_c\sum_{j=1}^{n_c}\Bigl[\max(0,c_j(x))\Bigr]^2 \label{defderhoc}
\end{equation}
leading to the function $f$ invoked in the problem statement \eqref{optproblem}. 
\e Notice that for real-life applications, requiring the full rigorous satisfaction of some of the state constraints might lead to infeasibility and hence unpredictable behavior of the underlying solver. Therefore, the use of soft constraints of the form \eqref{defderhoc} might be recommended even if one disposes of solvers that are able to handle non box-constrained problems. 
\subsection{The projection map}
\noindent The projection map on the admissible domain $[\underline x, \bar x]$ is denoted by $\Pi$, namely:
\begin{equation}
\Pi(x)= \begin{bmatrix} 
\pi_1\cr \vdots \cr \pi_n
\end{bmatrix}  \quad \text{s.t}\quad \pi_i=\min\Bigl(\bar x_i,\max\bigl(\underline x_i,x_i\bigr)\Bigr)\label{defdePi}
\end{equation}
\subsection{Sets of investigated steps}
\noindent When gradient-based solutions are investigated, the determination of the step size along the gradient descent direction and, in the case of the proposed algorithm, along the accelerating direction as in \eqref{FG2}, is a key issue. The following notation is used to parameterize the set of step sizes to be investigated as shown later on. As two different line search processes will be used, there are two different parameterizations that are discussed hereafter.
\e Given a triplet $(\underline{z},\bar z, n_g)\in \mathbb R^2\times \mathbb N$ such that $\underline z< \bar z$, the following two discrete sets are defined:
\begin{subequations}\label{defdesA}
\begin{align}
\mathbb A_\text{lin}(\underline z, \bar z,n_g)&:=\Bigl\{\underline z + \frac{i}{n_g-1}(\bar z-\underline z)\Bigr\}_{i=0}^{n_g-1} \label{Alin} \\
\mathbb A_\text{log}(\underline z, \bar z,n_g)&:=\Bigl\{10^z\Bigr\}_{z\in \mathbb A_\text{lin}(\underline z, \bar z,n_g)}  \label{Alog}
\end{align}
\end{subequations}
For instance\footnote{The presence of negative steps in $\mathbb A_\text{lin}$ is a posteriori justified by the results of the investigation contained in Section \ref{negative_c}.}:
\begin{subequations}
\begin{align}
\mathbb A_\text{lin}(-0.2, 1, 7)&:=\Bigl\{-0.2, 0, 0.2, \dots, 1.0\Bigr\}\label{AlinEx} \\
\mathbb A_\text{log}(-8, 0, 9)&:=\Bigl\{10^{-8},10^{-7},\dots,0, 1\Bigr\}  \label{AlogEx}
\end{align}
\end{subequations}
Anticipating the description of the algorithm that follows, these two families of sets will be used to perform respectively line search on the anti-gradient direction ($-g(y_k)$) and the momentum direction ($v_{k+1}-v_k$) invoked in \eqref{FG2}. The former involves the set $\mathbb A_\text{log}$ and the former involves the set $\mathbb A_\text{lin}$.
\e While, for a given $n_g$, a constant $\mathbb A_\text{lin}$ is used in all the numerical investigation to determine the step on the acceleration direction, the following transformations might be applied to the set $\mathbb A_\text{log}$ during the iterations which is used to determine the step along the gradient direction. These transformation introduce a kind of trust-region adaptation mechanism that is described in the following section 
\subsection{Trust-region adaptation to $\mathbb A_\text{log}$}
\noindent Depending on conditions that are expressed later on, one of the following transformations of the set of admissible steps $\mathbb A_\text{log}$ might be applied:
\begin{itemize}
\item The two-side contraction defined by:
\begin{equation}
\mathcal T_c\Bigl(A_\text{log}(\underline z, \bar z,n_g)\Bigr)=A_\text{log}(\underline z-1\, \bar z-1,n_g) \label{defdeLTc}
\end{equation}
which amounts at dividing all the elements by $10$,  reducing hence the size of candidate steps along the gradient\footnote{Notice however that a lower bound $z_\text{min}=10^{-16}$ is used in order to avoid $\underline z$ to totally vanish and froze future extension steps. }.
\item The one-side contraction defined by:
\begin{equation}
\mathcal T_r\Bigl(A_\text{log}(\underline z, \bar z,n_g)\Bigr)=A_\text{log}(\underline z, \bar z-\delta_z,n_g) \label{defdeLTr}
\end{equation}
where 
\begin{equation}
\delta_z=\rho\cdot (\bar z-\underline z) \quad;\quad \rho\in (0,1)\label{defdedeltaz}
\end{equation}
which amounts at reducing the larger step inducing a non uniform reduction among the set of eligible steps that keeps the smallest value unchanged. A typical value of $\rho$ is $\rho=0.05$ which is used in all the future investigations.
\item The expansion defined by:
\begin{equation}
\mathcal T_E\Bigl(A_\text{log}(\underline z, \bar z,n_g))\Bigr)=A_\text{log}(\underline z+\gamma\delta_z, \bar z+\delta_z,n_g) \label{defdeLTE}
\end{equation}
where $\delta_z$ is defined by the same expression \eqref{defdedeltaz} mentioned above. This operation corresponds to a \textit{non uniform} expansion of the gradient step candidate values. A typical value of $\gamma\in (0,1)$ is $\gamma=0.1$ which is used in all the future investigations.
\end{itemize}
\subsection{Updating rules}
\noindent Given a current iterate $\xi\in \mathbb R^n$, the two previous sets enable to define two sets of candidate updates along a direction $w\in \mathbb R^n$, namely:
\begin{subequations}
\begin{align}
\mathbb C_\text{log}(\xi\vert w,\eta)&:=\Bigl\{\Pi(\xi + \alpha w)\quad\vert\quad \alpha\in \bigl\{0, \eta\bigr\}\cup \mathbb A_\text{log} \Bigr\}\label{defdecandidatesetlog} \\
\mathbb C_\text{lin}(\xi\vert w)&:=\Bigl\{\Pi(\xi + c w)\quad\vert\quad c\in \bigl\{0\bigr\}\cup \mathbb A_\text{lin} \Bigr\}\label{defdecandidatesetlin}
\end{align}
\end{subequations}
The candidate values that minimize the cost function $f$ over the two above sets are denoted as follows ($\xi$ being the current value, the notation $\xi^+_\text{log}$ and $\xi^+_\text{lin}$ are associated to $\xi$):
\begin{subequations}
\begin{align}
\xi^+_\text{log}(w, \eta)&:= \text{arg}\min_{r\in \mathbb C_\text{log}(\xi\vert w,\eta)}\bigl[f(r)\bigr]\label{xiplulog}\\
\xi^+_\text{lin}(w)&:= \text{arg}\min_{r\in \mathbb C_\text{lin}(\xi\vert w)}\bigl[f(r)\bigr]\label{xipluslin}
\end{align}
\end{subequations}
Notice that the optimization problems \eqref{xiplulog} and \eqref{xipluslin} are defined on discrete sets of low cardinality\footnote{Typically, $n_g\in \{3,5,8,10\}$.} $n_g+2$ and $n_g+1$ respectively and will be solved by simple enumeration. 
\e Notice that the use of $\eta$ in \eqref{defdecandidatesetlog} and \eqref{xiplulog} guarantees that the step size $\eta$ remains an available option regardless of the transformation that $\mathbb A_\text{lin}$ might face during the iteration. Similarly, the insertion of $0$ inside $\mathcal C_\text{lin}$ enables to ensure that the  \textit{no acceleration} option is always available at the acceleration step as it is shown later on. 
\e By now we have all we need to move to the presentation of the proposed algorithm.
\section{The Search And Accelerate (\textbf{SaA}) algorithm}\label{sec-algo}
\subsection{Brief presentation of the algorithm}
\noindent As in the original \textbf{FGM} defined by \eqref{FGeq}, the algorithm processes two internal states, namely: 
\begin{itemize}
\item $y_k\in \mathbb R^n$ represents the \textit{official} iterate that represents the intermediate version of the solution to \eqref{optproblem} at iteration $k$. This is obtained after an acceleration step as shown below.
\item $v_k\in \mathbb R^n$ is the iterate that is computed through a step on the gradient line before the  acceleration step that leads to $y_k$. 
\end{itemize} 
The description of the algorithm is all about explaining how it updates the pair $(v_k,y_k)$. This together with some stopping conditions enables repeated application of the updating process over a finite number of iterations, say $N$, leading to the returned solution $y_N$.
\subsection{Detailed description of the steps}
\noindent \textbf{Initialization}. First of all the pair $(v_0, y_0)$ is initialized based on some provided initial guess $x_0\in (\underline x,\bar x)$:
\begin{equation}
k=0,\quad v_k=y_k=x_0\quad;\quad \eta=10^{-16}\label{initalization}
\end{equation}
The integer $n_g$ is chosen which enables to initialize the sets $\mathbb A_\text{lin}$ and $\mathbb A_\text{log}$ according to:
\begin{equation}
\mathbb A_\text{lin}:= \mathbb A_\text{lin}(-0.2,0,n_g)\quad;\quad \mathbb A_\text{log}:= \mathbb A_\text{log}(-8,0,n_g)\label{initaliazinglesAmathbb}
\end{equation}
\textbf{Gradient computation}. The gradient is computed at the current solution using $g_k= g(y_k)$\e 
\textbf{Gradient step}. Using the definition \eqref{xiplulog}, a candidate value $\hat v_{k+1}$ to  $v_{k+1}$ is computed:
\begin{equation}
\hat v_{k+1} = \bigl[y_k\bigr]^+_\text{log}(-g_k, \eta) \label{hatvk+1}
\end{equation}
The relevance of $\hat v_{k+1}$ and the updating it should induce depend on which one of the following possibilities holds: \\
\begin{enumerate}
    \item Either $f(\hat v_{k+1})\ge f(y_k)$ which might be due to smallest step size being too high or simply to the projection operator. In this case a two-side contraction is triggered [see \eqref{defdeLTc}] that updates the set of eligible steps: 
\begin{equation}
\mathbb A_\text{log} \leftarrow \mathcal T_c(\mathbb A_\text{log}) \quad;\quad v_{k+1}\leftarrow v_k\label{contraction}
\end{equation}
the last updating means that the candidate value $\hat v_{k+1}$ is rejected and the iteration stops with only $\mathbb A_\text{log}$ being modified. \\
\item Or $f(\hat v_{k+1})<f(y_k)$ in which case, the update is adopted but two possibilities are to be distinguished as they lead to different updating rules for $\mathbb A_\text{log}$: \\
\begin{enumerate}
    \item If the best value corresponds to the largest update, then an extension [see \eqref{defdeLTE}] is applied to $\mathbb A_\text{log}$: 
    \begin{equation}
\mathbb A_\text{log} \leftarrow \mathcal T_E(\mathbb A_\text{log}) \quad;\quad v_{k+1}\leftarrow \hat v_{k+1}\label{extension}
\end{equation}
\item Otherwise, the best step is an intermediate one and a one-side contraction defined by \eqref{defdeLTr} is applied  as the minimal value is small enough (in order to tighten the interval around the successful interior step size):
    \begin{equation}
\mathbb A_\text{log} \leftarrow \mathcal T_r(\mathbb A_\text{log}) \quad;\quad v_{k+1}\leftarrow \hat v_{k+1}\label{reduction}
\end{equation}
\end{enumerate}
\end{enumerate}
The whole gradient step of the \textbf{SaA} algorithm discussed above is sketched on Figure \ref{fig:gradient-step}.
\begin{figure}[h!]
    \centering
    \framebox{\includegraphics[width=0.99\linewidth]{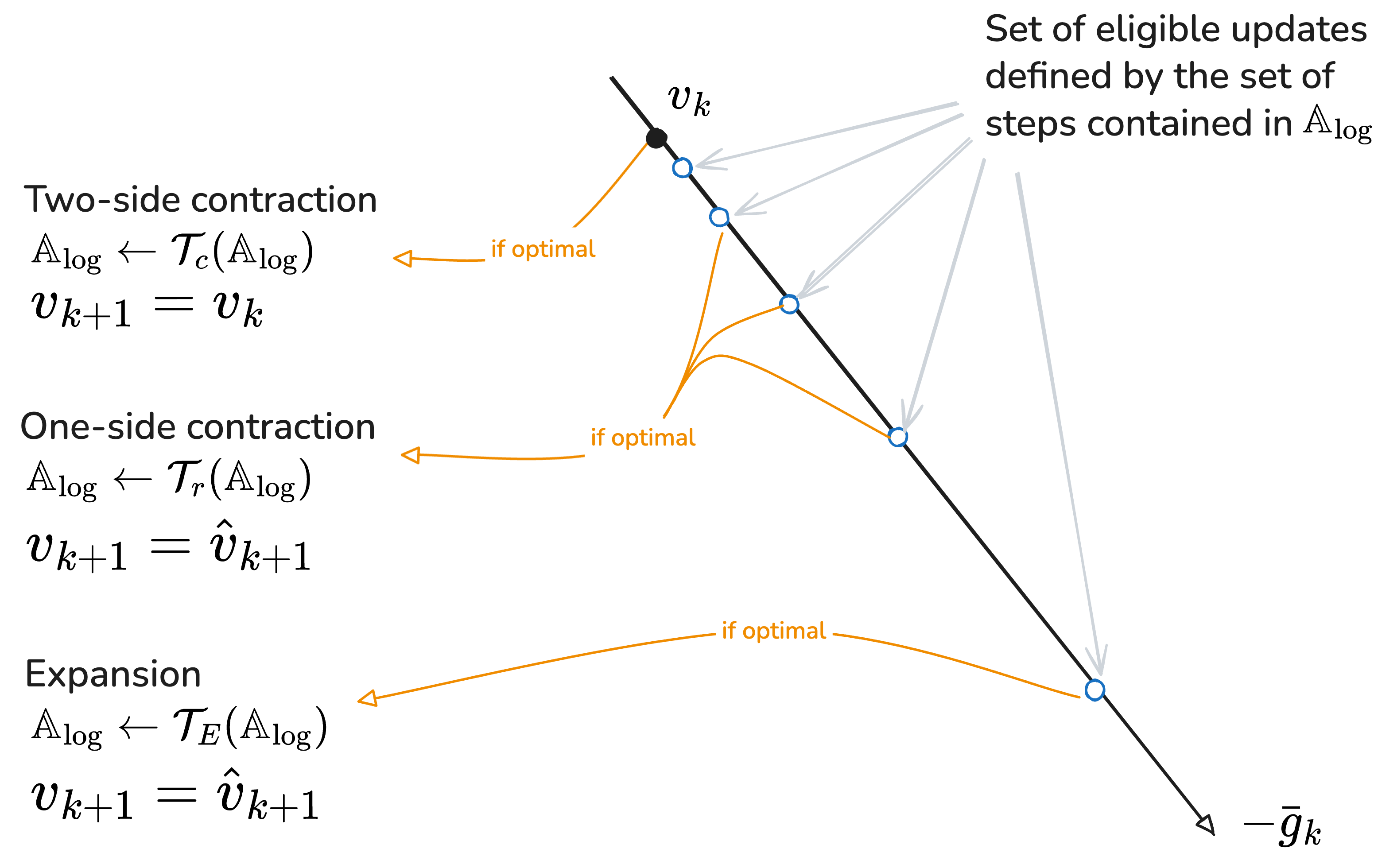}}
    \caption{Sketch of the \textbf{gradient step}. Depending on which element of $\mathbb C_\text{log}(v_k\  \vert -g_k,\eta)$ is optimal, the computed candidate update $\hat v_{k+1}$ is adopted or not and a specific transformation of the set of steps $\mathbb A_\text{log}$ is applied.}
    \label{fig:gradient-step}
\end{figure}
\e 
Recall that if the gradient step is unsuccessful, the current iteration stops, returns the updated set $\mathbb A_\text{log}$. Otherwise the iteration continue with the acceleration step described below. 
\e \textbf{Acceleration step}. The objective of this step is to seek a better iterate that lies on the line defined by $v_{k+1}-v_k$ which is different from zero since as explained earlier, this step is fired only in case of successful gradient step and hence $v_{k+1}\neq v_k$. More precisely, the cost function is evaluated at the elements in $\mathbb C_\text{lin}(v_{k+1}\ \vert\ v_{k+1}-v_k)$ defined by \eqref{defdecandidatesetlin}. Notice that all these elements lie in the admissible domain by virtue of its very definition.
\e Consequently, the acceleration update is given by:
\begin{equation}
y_{k+1} \leftarrow \bigl[v_{k+1}\bigr]^+_\text{lin}(v_{k+1}-v_k)\label{accelerationstep}
\end{equation}
which fully defines the acceleration step that is summarized in Figure \ref{fig:acceleration}. A formal version of the whole algorithm including the gradient and the acceleration steps is also shown in Algorithm \ref{algo}
\begin{figure}[h!]
    \centering
    \framebox{\includegraphics[width=0.85\linewidth]{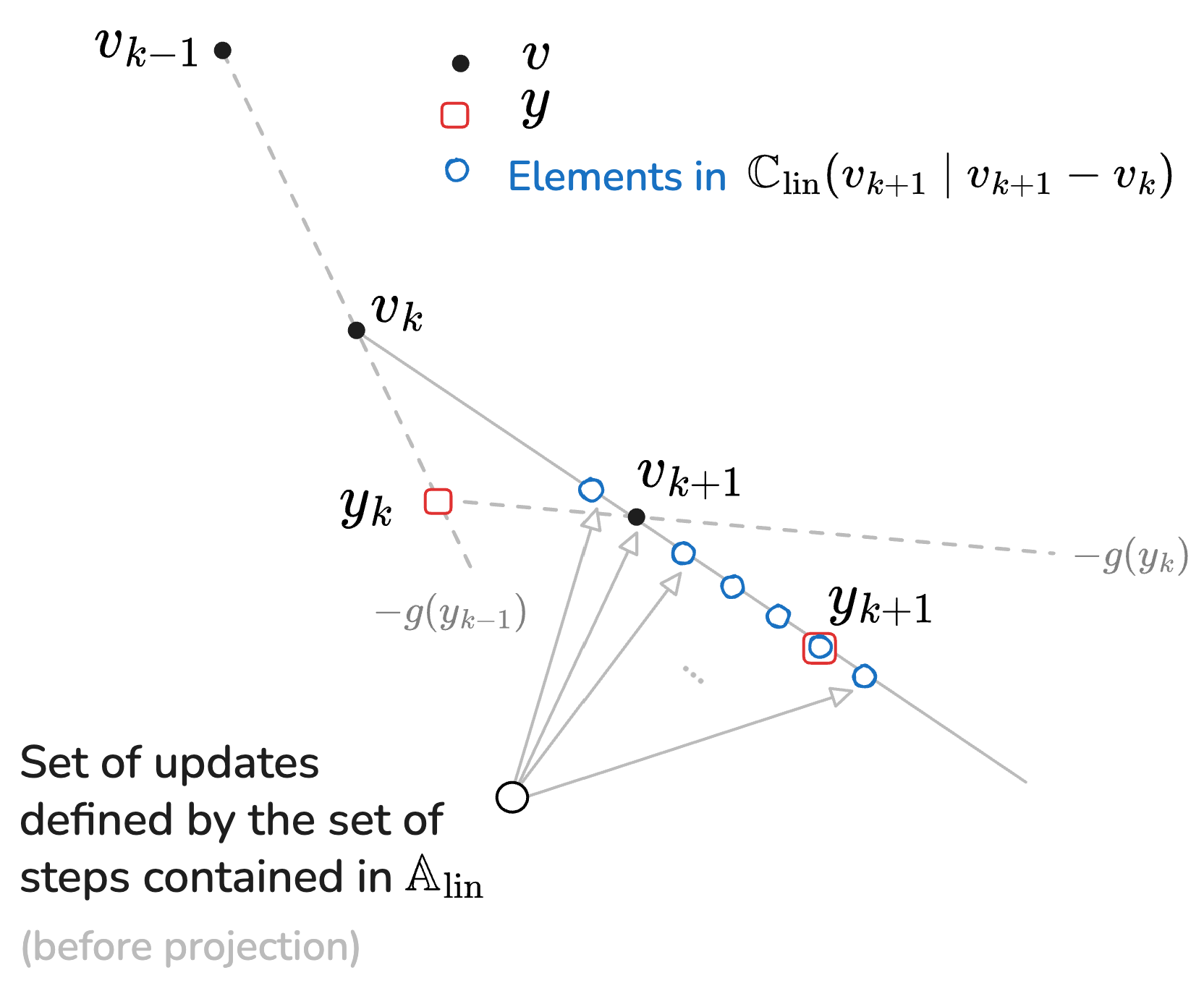}}
    \caption{Sketch of the \textbf{acceleration step}. After two successful gradient steps leading to $v_k$ and $v_{k+1}$, an attempt is made to find a better admissible solution along the direction $v_{k+1}-v_k$. This is done by evaluating the cost function at the projected elements of $\mathbb C_\text{lin}(v_{k+1}\ \vert\ v_{k+1}-v_k)$. (The projection is not shown in the illustration). The presence of possible negative steps (represented by the point lying between $v_{k}$ and $v_{k+1}$) is a posteriori justified by the results of the investigation contained in Section \ref{negative_c}}
    \label{fig:acceleration}
\end{figure}
\begin{algorithm}
\footnotesize
\caption{Search \& Accelerate (\textbf{SaA})}\label{algo}
\begin{algorithmic}[1]  
\Statex \textbf{Input parameters} (with suggested default values): 
\Statex $n_g=5$, $\eta=10^{-16}$, \texttt{maxIter}=200, \texttt{epsG}=$10^{-8}$, $\underline x$, $\bar x$, $x_0$
\Statex \textbf{Initialization}: 
    \Statex $\mathbb A_\text{lin}\leftarrow A_\text{lin}(-0.2,1,n_g)$ \Comment{\scriptsize {\color{Gray} Acceleration search grid \eqref{Alin}}}
    \Statex $\mathbb A_\text{log}\leftarrow A_\text{lin}(-8,1,n_g)$ \Comment{\scriptsize {\color{Gray} Gradient search grid \eqref{Alog}}}
    \Statex $k=0, v_0=y_0=x_0$, \texttt{normG}=$\infty$ \Comment{\scriptsize {\color{Gray} Initialize the internal states}}
    \Statex \hrulefill
    \While{($k\le \texttt{maxIter}$) and (\texttt{normG}$>$\texttt{epsG})}
        \Statex {\color{Gray} ----- The Gradient Step -----}
        \State $g_k\leftarrow g(y_k)$\Comment{{\scriptsize\color{Gray} Evaluate gradient}}
        \State $\hat v_{k+1} \leftarrow\bigl[y_k\bigr]^+_\text{log}(-g_k, \eta)$\Comment{\scriptsize {\scriptsize\color{Gray} Gradient search candidate based on $\mathbb A_\text{log}$}}
        \State $\hat f_{k+1}\leftarrow f(\hat v_{k+1})$\Comment{{\scriptsize\color{Gray} Evaluate cost at the candidate value}}
        \If{$\hat f_{k+1}==f(y_k)$} \Comment{\scriptsize {\color{Gray} $f(y_k)$ is already computed to determine $\hat v_{k+1}$}}
        \State $\mathbb A_\text{log} \leftarrow \mathcal T_c(\mathbb A_\text{log})$ \Comment{\scriptsize {\color{Gray} Two side contraction \eqref{defdeLTc}}}
        \State $(v_{k+1},y_{k+1})\leftarrow (v_k,y_k)$ \Comment{\scriptsize {\color{Gray} Keep the same solution}}
        \State $k\leftarrow k+1$
        \State\texttt{Goto} 3: 
        \Else{}
        \State $v_{k+1}\leftarrow \hat v_{k+1}$ \Comment{\scriptsize {\color{Gray} Adopt the new solution}}
        \If{($\hat v_{k+1}$ results from the last item in $\mathbb C_\text{log}(y_k,g_k,\eta)$)} \Comment{\scriptsize {\color{Gray} see \eqref{defdecandidatesetlog}}}
        \State $\mathbb A_\text{log} \leftarrow \mathcal T_E(\mathbb A_\text{log})$ \Comment{\scriptsize {\color{Gray} Extension}}
        \Else{}
        \State $\mathbb A_\text{log} \leftarrow \mathcal T_r(\mathbb A_\text{log})$\Comment{\scriptsize {\color{Gray} One-side contraction}}
        \EndIf
        \EndIf
        \Statex {\color{Gray}----- The acceleration Step -----}
        \State $y_{k+1} \leftarrow \bigl[v_{k+1}\bigr]^+_\text{lin}(v_{k+1}-v_k)$ \Comment{\scriptsize {\color{Gray} Update $y$ along the acceleration path [see \eqref{accelerationstep}]}}
        \State $k\leftarrow k+1$
    \EndWhile
\end{algorithmic}
\end{algorithm}
\subsection{Analysis of the algorithm's properties}
\noindent The first statement concerns the complexity of the algorithm, namely:
\begin{prop}[Complexity]
A single iteration of the algorithm involves:\\
- A single evaluation of the gradient.\\
- At most $2n_g+3$ evaluations of the cost function.
\end{prop}
{\sc Proof.} Indeed, the gradient is only evaluated at step 2: then $n_g+2$ evaluations of the cost function are needed\footnote{This is needed to evaluate along the elements defined by $\mathbb A_\text{log}$ plus the two values in $\{0,\eta\}$ in \eqref{defdecandidatesetlog}.} in step 3: in order to find the best elements within $\mathbb C_\text{log}$. Depending on the success/failure of this step, either additional $n_g+1$ evaluations\footnote{This is needed to evaluate along the elements defined by $\mathbb A_\text{lin}$ plus the one involving $0$ in \eqref{defdecandidatesetlin}.} are performed in step 18: or the iteration ended by simply updating the search set parameter in $\mathbb A_\text{log}$. $\hfill \Box$
\e Before we state the main convergence result, we need to recall the following preliminary result \cite{BeckTeboulle2009}:
\begin{lem}[Decrease of the projected gradient update]\label{lem1}\ \\
For a cost function with Lipchitz constant $L$, when the step size $\alpha$ satisfies $\alpha< 1/(2L)$, then the resulting projected update:
\begin{equation}
x_{k+1}=\Pi\Bigl(x_k - \alpha   g(x_k)\Bigr)\label{newx}
\end{equation}
satisfies the following inequality:
\begin{equation}
f(x_{k+1})\le f(x_k)-\Bigl[1-\alpha L/2\Bigr]\|x_{k+1}-x_k\|^2\label{ineq}
\end{equation}
Moreover, if $x^\star$ is a stationary solution to \eqref{newx} then $x^\star$ meets the KKT optimality conditions for the optimization problem \eqref{optproblem}.
\end{lem}
{\sc Proof}. Inequality \eqref{ineq} is a direct result of the condition $\alpha<1/(2L)$ and the very definition of the Lipchitz constant. As for the second statement, see \cite{BeckTeboulle2009}$. \hfill \Box$
\begin{prop}[Convergence]
If the following conditions hold \\
1) The cost function $f$ admits a Lipchitz constant $L$, \\
2) The cost function is continuous and bounded below,\\
3) The parameter $\eta$ involved in \eqref{defdecandidatesetlog} satisfies\footnote{This condition might be satisfied by taking very small value of $\eta$ without inducing slow convergence as discussed in section \ref{sec:discussion}.}:
\begin{equation}
\eta< \frac{1}{2L}, \label{condition.}
\end{equation}
then algorithm \ref{algo}, when allowed to perform an unlimited number of iterations, converges to a point that meets the KKT optimality conditions for the optimization problem \eqref{optproblem}. 
\end{prop}
{\sc Proof}. Indeed, each time the gradient step is executed, the optimal solution $\hat v_{k+1}$ shows necessarily a better or equal cost than the one corresponding to the element associated to $\eta$, say $v^{[\eta]}_{k+1}$  in \eqref{defdecandidatesetlog}. But this candidate value associated to $\eta$ is precisely defined by:
\begin{equation}
v^{[\eta]}_{k+1} = \Pi(y_k-\eta g(y_k)) \label{rgf1}
\end{equation}
and since by assumption, $\eta<1/(2L)$, inequality \eqref{ineq} of Lemma \ref{lem1} enables to write (using the notation $\beta:=(1-\alpha L/2)$):
\begin{align}
f(v^{[\eta]}_{k+1}) &\le f(y_k)-\beta\times\|v^{[\eta]}_{k+1}-y_k\|^2\label{jhgtt7}\\
&\le f(y_k)-\beta\times \|\Pi(y_k-\eta g(y_k))-y_k\|^2\label{jhgtt8}
\end{align}
and hence:
\begin{equation}
f(\hat v_{k+1})\le f(y_k)-\beta\times \|\Pi(y_k+\eta g(y_k))-y_k\|^2 \label{kjjg9}
\end{equation}
From here two situations should be considered:
\vskip 2mm\noindent 
1) Either the condition of step 5:, namely $f(\hat v_{k_0+1})=f(y_{k_0})$, is satisfied at some iteration $k_0$, in which case, \eqref{kjjg9} implies that 
\begin{equation}
0=\Pi(y_{k_0}+\eta g(y_{k_0}))-y_{k_0} \label{gftg}
\end{equation}
making $y_{k_0}$ a stationary solution to \eqref{newx} and hence satisfying the KKT condition. 
\vskip 2mm\noindent 
2) Or the algorithm always triggers the branch defined at step 10: leading to $v_{k+1}=\hat v_{k+1}$ and injecting this in \eqref{kjjg9} leads to:
\begin{equation}
f(v_{k+1})\le f(y_k)-\beta\times \|\Pi(y_k+\eta g(y_k))-y_k\|^2 \label{jnbg76}
\end{equation}
and since the acceleration step examine $(n_g+1)$ options inside $\mathbb A_\text{lin}$ that include $v_{k+1}$, thanks to the insertion of $0$ in \eqref{defdecandidatesetlin}, it comes that the finally adopted $y_{k+1}$ is such that:
\begin{equation}
f(y_{k+1})\le  f(v_{k+1})\label{kjj}
\end{equation}
which together with \eqref{jnbg76} implies:
\begin{equation}
f(y_{k+1}) \le f(y_k)-\beta\times \|\Pi(y_k+\eta g(y_k))-y_k\|^2\label{enfin}
\end{equation}
Now since the cost function $f$ is continuous and bounded below, the last inequality means that the sequence $\{y_k\}$ converges to a stationary solution of \eqref{newx} and hence asymptotically meets the KKT conditions which ends the proof. $\hfill \Box$
\subsection{Discussion}\label{sec:discussion}
\noindent The aim of this discussion is to highlight some details of the implementation that enable to overcome the limitations underlined in the introduction and leading to enhancing both performance and stability of the iterations. 
\begin{itemize}
\item[$\checkmark$] Notice first of all the role played by the inclusion of an extremely small\footnote{Default value is $\eta=10^{-16}$.} $\eta$ in $\mathbb C_\text{log}$ which enabled through the equations \eqref{rgf1}-\eqref{kjjg9} to derive the inequality needed for the proof of stability.
\item[$\checkmark$] Contrary to traditional approaches involving constant gradient steps, the use of extremely small values for $\eta$ does not induce slow convergence rates thanks to the fact that the other options in $\mathbb C_\text{log}$ are systematically \textit{investigated} (in step 3:) through function evaluations over the elements of the set $\mathbb C_\text{log}(y_k\vert -g(y_k),\eta)$ enabling fast decrease (when possible) while using very small values of $\eta$. 
\item[$\checkmark$] While traditionally, the acceleration is avoided when the gradient step is dynamically updated, the inclusion of several acceleration gains (including $0$) is here used to get the inequality \eqref{kjj} that guarantees the overall decrease in an acceleration-enabled context. This enables to combine the search for better decrease in the gradient step while allowing for the acceleration step to be examined. 
\item[$\checkmark$] Contrary to the \textbf{FGM} method, the previous analysis of the algorithm's properties does not include formal results regarding the convergence rate. While this might be attempted, it is our belief that what does really matter in NMPC is the behavior of the algorithm during the first iterations, namely, the iterations that last less than the control updating period. The results shown in the next section suggests a quite impressive experimentally assessed convergence rate. It might even be conjectured that the will to get a provable convergence rate of $1/k^2$ for the \textbf{FGM} lies behind not using the so simple ideas proposed in the current contribution which leads to much better results with guaranteed stability.
\end{itemize}
The objective of the next section is to examine and compare such a behavior among algorithms belonging to few different families and when used to address a high number of randomly generated problems that belong to an existing and available benchmark of problems that are described hereafter. 
\section{The benchmark}\label{sec-bench}
\noindent The benchmark used here in the evaluation of the algorithm is made publicly available in a \textsc{Kaggle} repository \cite{mazen_alamir_2026_kaggle_bench_opt}. It consists in 600 optimization problems\footnote{The repository enables to download the list of problems as python objects that export two methods representing the cost function and the gradient.} of the form:
\begin{equation}
\min_{x\in \mathbb X} f(x):=\bigl\vert P(x)-P(x^\star)\bigl\vert^{2m}\quad;\quad \mathbb X:=[-5,+5]^n\label{defdepbkaggle}
\end{equation}
where $P$ is a multivariate polynomial in the variable $x\in \mathbb R^{n}$. Both $P$ and $x^\star$ take different value depending on the problem and are randomly generated. The degrees of the polynomial takes values inside $\texttt{deg}\in \{1,2,3\}$, the integer $m\in \{1,2\}$. This means that the cost functions are polynomial of degrees that lie inside $\{2,4,\dots,12\}$. The dimension of the decision variable takes different values ranging from $n=2$ to $n=1000$, more precisely, 10 different values of $n$ are used inside the following set:
\begin{equation}
n\in \{2, 5, 10, 15, 20, 50, 100, 200, 500, 1000\}\label{defdesnkaggle}
\end{equation} 
For each of the 60 values of the triplet $(\texttt{deg},n,m)$, 10 random problems are generated. For Five of these 10 problems, the value of $x^\star$ is randomly generated \textbf{inside} the admissible set $\mathbb X$ while in the remaining 5, $x^\star$ is forced to lie outside the hyper cube. By so doing, we force the optimal value of the cost function to be equal to $0$ for the first group of problems while this is not guaranteed for the second so that many instances of the problem correspond to optimal solutions that are not stationary points for the gradient. 
\e Notice that given the relatively high degrees of the polynomials involved in the cost functions, there might be several minima with zero value that lie inside the admissible hyper-cube $[-5,+5]^n$ so we focus on the value of the cost function rather than on how the solutions found by the algorithm relate to the specific value of the parameter $x^\star$ for each problem. 
\e In order to get comparable and reproducible results, the same initial guess $x_0=(0.1,\dots,0.1)^T$ is used for all the algorithm and all the instances of the problem.
\section{Numerical investigation}\label{sec-results}
\noindent Let us start by stating the precise settings for the different solvers.
\subsection{Settings for alternative solvers}\label{sec:settings}
\noindent As previously explained, the ability of a solver to operate with a unique setting for all possible problems is mandatory should it be applied to NMPC. That is the reason why comparisons are drawn between solvers that keep their setting unchanged over the whole set of 600 optimization problems included in the benchmark. 
\e Nevertheless, for a fair comparison, different such \textit{unchanged} settings might be used in order to show the extent to which the setting impacts the comparison. 
\subsubsection{Settings for \textbf{FGM}}
\noindent For this solver, the only possible parameter is the gradient step size $\alpha$ (which is expected to be lower than the unknown $1/L$. The $\beta_k$ gain used in the acceleration step is defined following the rules defined in \cite{nesterovfg2004} which is supposed to induce stability provided that the previous condition on $\alpha$ is satisfied. In order to show the impact of this choice, three values of $\alpha$ are investigated, namely $$\alpha\in \{10^{-5}, 10^{-10}, 10^{-16}\}$$ as these values enable to explore fast but risky settings as well as safe but potentially inefficient settings.
\subsubsection{Settings for \textbf{FISTA}}
\noindent Recall that \textbf{FISTA} implements an estimation of $L$ through a backtracking process. More precisely, an initial estimate $L_0$ is initialized and the process amounts at increasing the value of the current estimation each time a failure in decrease is encountered. The expansion rate is set to $2$ while two different values of the initial guess $L_0$ are investigated, namely: $$L_0\in \{10^{-6}, 1, 10^6\}$$ 
Notice that \textbf{FISTA} implements an acceleration step that is the same as in the \textbf{FGM}.
\subsubsection{Settings for \textbf{Strong Armijo-Wolf}}
\noindent This method does not implement an acceleration step but uses more involved line search that aims at increasing the efficiency of the gradient step. As explained earlier, this involves the enforcement of the decrease of the cost function and the decrease of gradient norm  (curvature condition). These two conditions involve two constants $c_1$ and $c_2$ which are taken respectively equal to the default values recommended in \cite{MoreThuente1994, HagerZhang2005}, namely:
$$c_1=10^{-4}\quad;\quad c_2=0.9$$
The initial value of $\alpha$ is taken equal to $1$ and the adaptation gain $\beta=0.5$ is used. As this method is an adaptive method, the impact of these parameters on the results is minor (and it appeared to be so in the experiments). Therefore, the unique setting described by the above values is used for this method in order to avoid overloading the comparison figure's below. 
\subsubsection{Settings for the proposed \textbf{SaA} algorithm}
\noindent As far as the proposed algorithm is concerned, only the number $n_g$ of search values inside the set of steps $\mathbb A_\text{log}$ and $\mathbb A_\text{lin}$ is explored. More precisely, the following values are tested: 
$$n_g\in \{3, 8, 20\}$$
each of which is kept constant over all the problems. All the other parameters of the algorithms take the default value shown in the description of the Algorithm [see the initialization step of Algorithm \ref{algo}).
\subsection{Results on the benchmark's problems}
\noindent In the forthcoming experiments, a maximum number of iterations is fixed to \texttt{maxIter}=200 with an early stopping criterion associated to the norm of the gradient being lower than $10^{-8}$.
\subsubsection{Cost achieved with $\texttt{maxIter}=200$}
\noindent Figure \ref{fig:finalCost} shows the statistics (percentiles) of the final achieved cost within the allowed 200 iterations for the different pairs of (algorithm/settings) described in Section \ref{sec:settings}.\e 
\begin{figure}[h!]
    \centering
    \includegraphics[width=\linewidth]{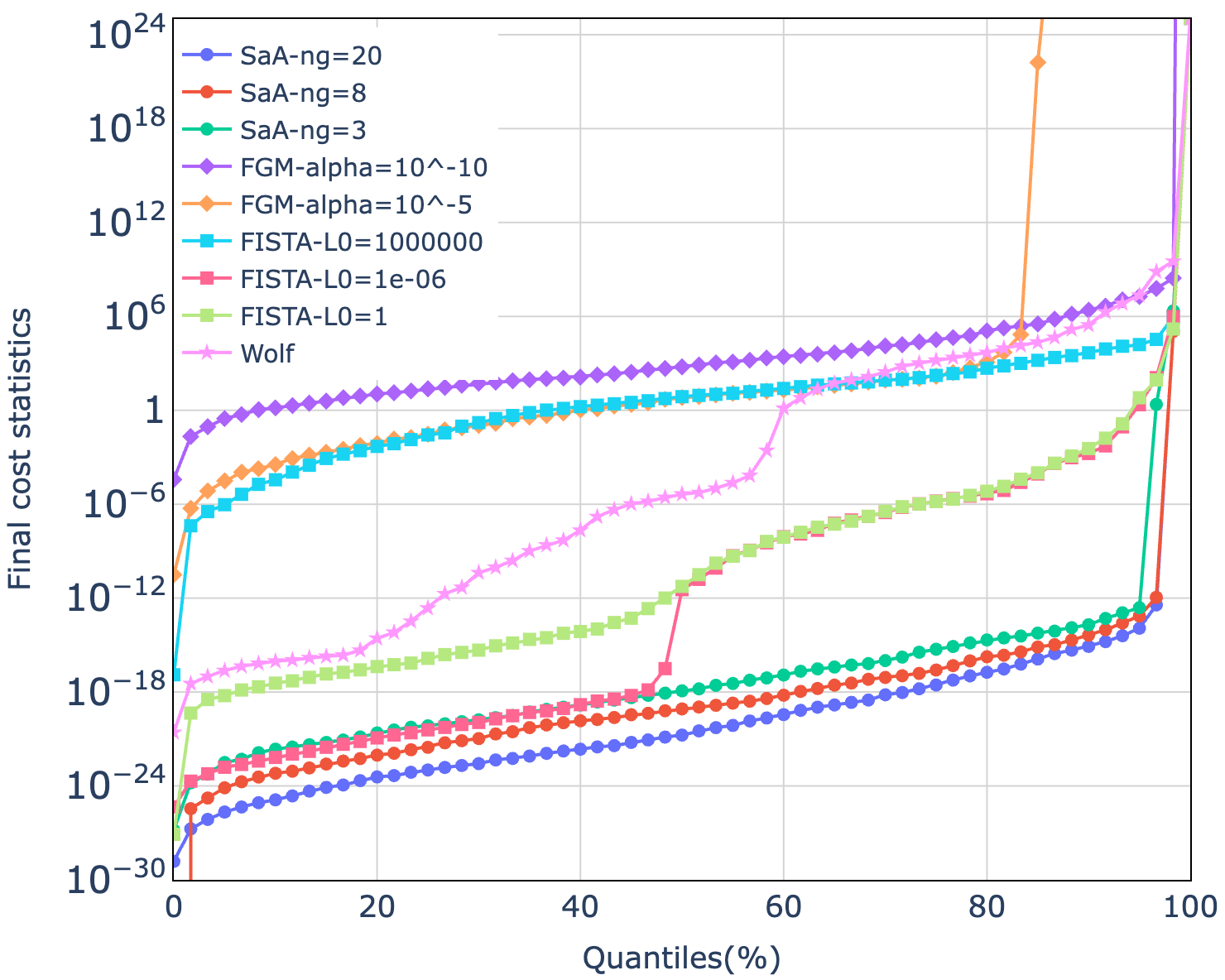}
    \caption{statistics (percentiles) of the final cost achieved within the allowed 200 iterations for the different pairs of (algorithm/settings) described in Section \ref{sec:settings}.}
    \label{fig:finalCost}
\end{figure}
This figure suggests the following comments: 
\begin{itemize}
\item First of all the \textbf{FGM} with high $\alpha$ ($10^{-5}$) diverges for a non negligible part of the set of problems while the one with much lower $\alpha$ ($10^{-10})$ shows, as expected, rather bad achieved cost resulting from the slow convergence and the limited number of iterations. This is precisely the unsolvable dilemma that is associated to the choice of a constant (non adaptable) $\alpha$ that has been already discussed earlier in the paper. Given the bad performance of \textbf{FGM} with $\alpha=10^{-10}$, the results using $\alpha=10^{-16}$ were omitted in the figure.
\item Regarding \textbf{FISTA}, the use of backtracking penalizes the use of high values of initial $L_0$ inducing slow convergence. Smaller values give better results which remains highly outperformed by the \textbf{SaA} in all its settings. This is most likely due to the absence of search in the acceleration step. This results in orders of magnitude in the final achieved cost over almost 40\% of the problems included in the benchmark. 
\item Regarding the \textbf{Srong Armijo-Wolf} line-search based algorithm, it clearly lies between the \textbf{FGM} and the \textbf{FISTA} with quite bad final achieved cost over $40\%$ of the problems.
\item The proposed \textbf{SaA} clearly outperforms the other alternatives and inside the set of settings for \textbf{SaA} and as far as the final cost within the 200 allowed iterations is considered, the higher $n_g$ the lowest is the final achieved cost which is no surprise as long as the associated computation cost is not investigated yet. 
\end{itemize}
It is worth underlying that the highest values of the final cost (around $10^{25}$ might suggest that the iteration diverges for some of the problems and this even for \textbf{SaA} but this is not the case as it is shown in Figure \ref{fig:ratios}. Indeed this figure shows the ratio between the finally achieved cost and the initial value of the cost for each of the 600 problems. This figure shows that except for one problem where the final cost is infinitely close to the initial one (contraction lower than 0.99), all the instances of the problem correspond to a (rather drastic) contraction in accordance with the convergence result.
\begin{figure}[h!]
    \centering
    \includegraphics[width=\linewidth]{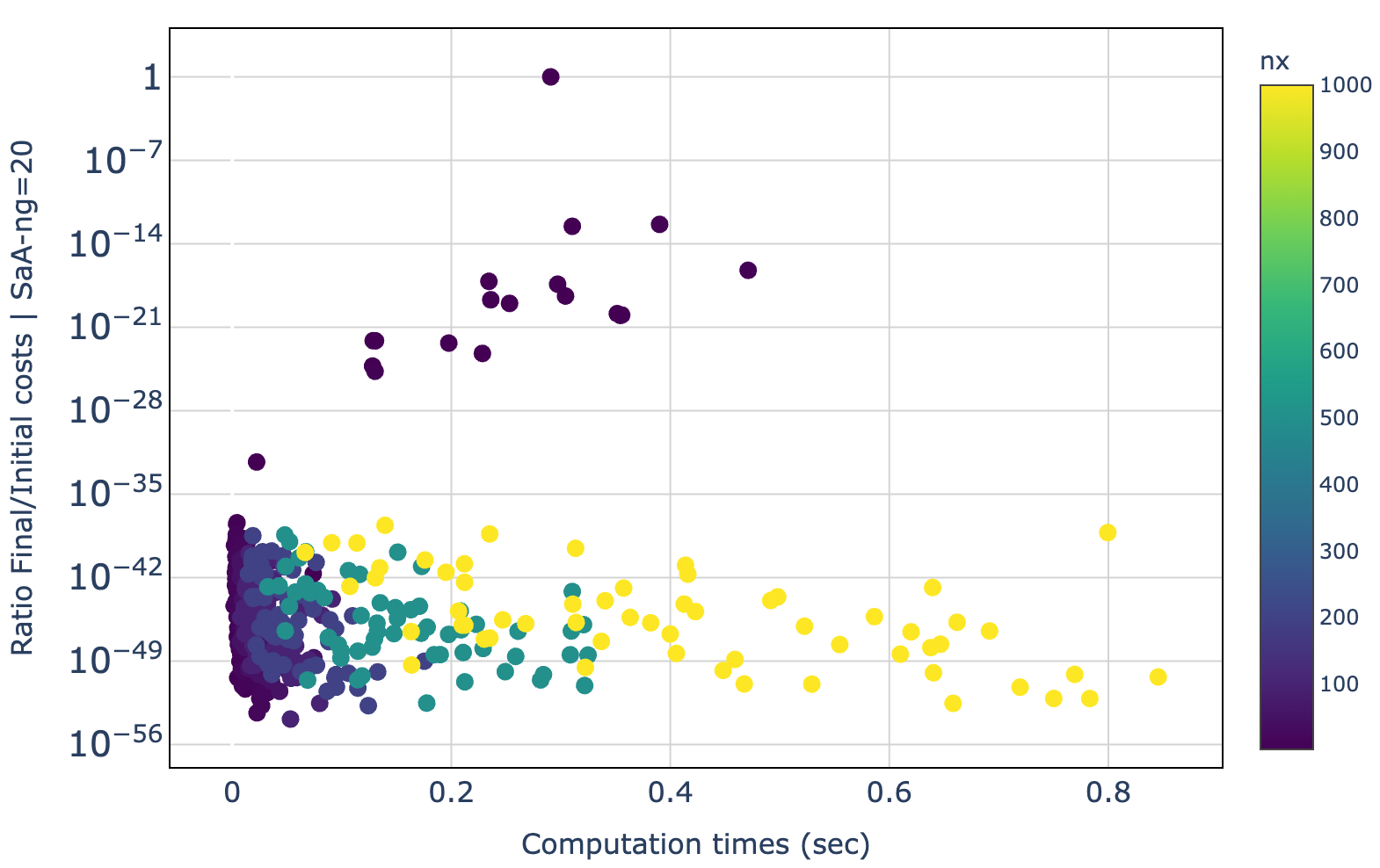}
    \caption{The ratios between the finally achieved cost and the initial value of the cost for each of the 600 problems. Notice that the colors refer to the dimension of the decision variables. }
    \label{fig:ratios}
\end{figure}
\e 
\noindent In order to explain the very high values of the final cost for a tiny fraction of the problems as shown in Figure \ref{fig:finalCost}, it is worth remembering that while the admissible set is $[-5,5]^n$, the vector $x^\star$ involved in \eqref{defdepbkaggle} is not necessarily inside this set. This with the fact that, as it is explained in Section \ref{sec-bench}, some of the optimization problems shows a cost function that is a polynomial of degree 12 (this is the case when \texttt{deg}=3 and $m=2$) lead to few problems for which the algorithm reaches rapidly a KKT point that lies on the boundary of the admissible hypercube and which corresponds to a cost function value that is extremely close to the value at the initial guess. 
\e 
From the results of Figure \ref{fig:finalCost}, it is fair to conclude that \textbf{SaA} achieves the best costs  among almost all the 600 optimization problems. Actually the returned solutions represent the optimal cost for more than 97\% of the problems as the optimal returned values are numerically $\approx$ 0, for the remaining problems it can be conjectured that the returned solutions which are not vanishing (because when $x^\star$ lies outside the admissible domain), the truly optimal value may not be $0$. Figure \ref{fig:ratios} shows that for all the problems, except one, the contraction of the cost function value is better than $10^{-12}$.
\e This being said, we still need to investigate whether this performance in terms of the final cost is obtained at the cost of larger computation times.
\subsubsection{Computation time}

\begin{figure}[h!]
    \centering
    \includegraphics[width=\linewidth]{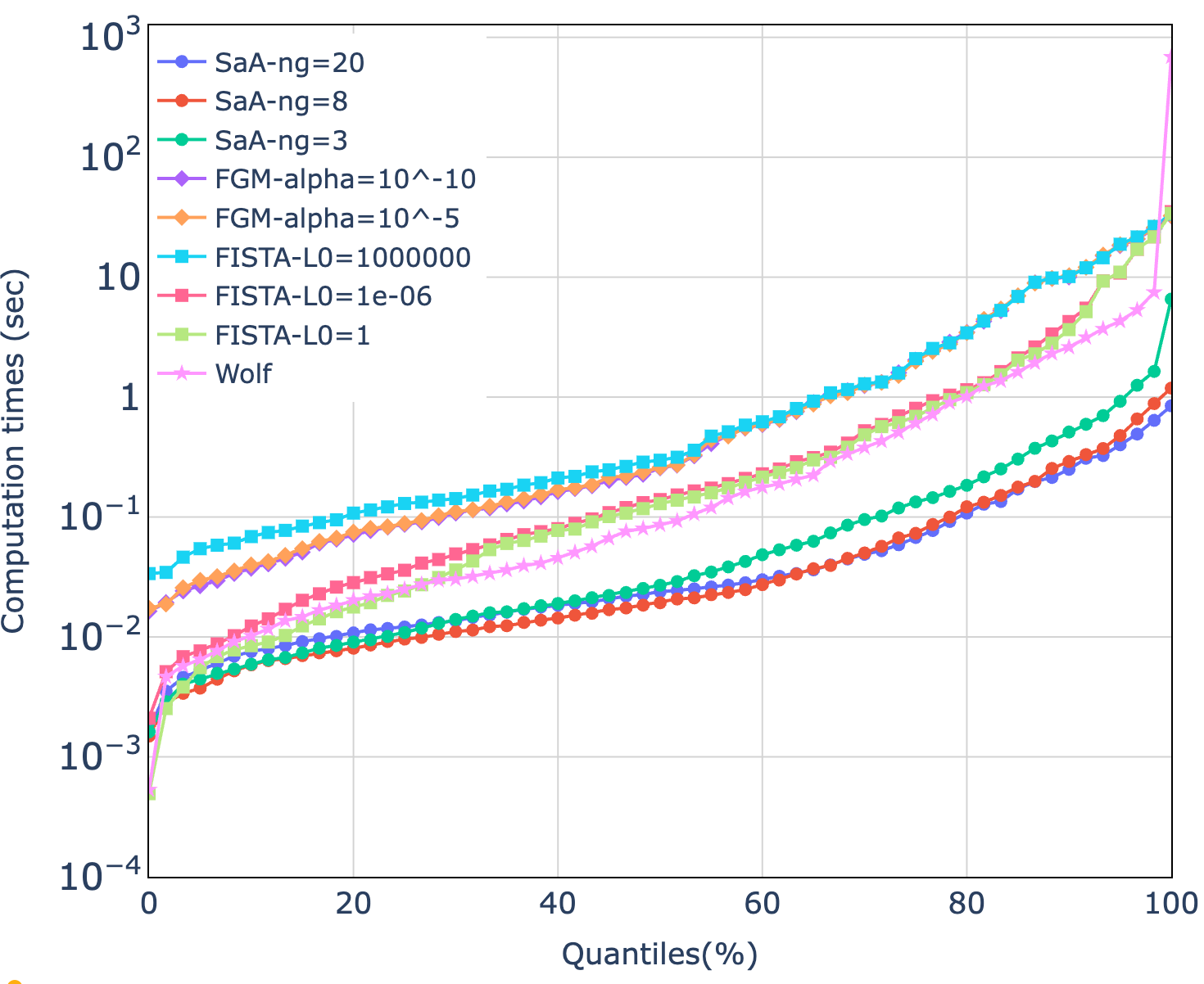}
    \caption{Statistics of the computation times when solving the optimization problems included in the benchmark using the different pairs of (algorithm/setting).}
    \label{fig:cpu}
\end{figure}
\noindent Figure \ref{fig:cpu} shows the statistics of the computation times when solving the optimization problems included in the benchmark. This figure shows that for almost all the problems, the proposed algorithm (\textbf{SaA}), regardless of its settings (value of $n_g$) achieves the best cost and simultaneously needs less time than all the other alternatives. For relatively large computation times, the \textbf{SaA} is more than 10 times faster than the closest of the other solvers.
\e A noticeable fact is that when examining the different settings of \textbf{SaA}, the computation time does not necessarily increase when the number of exploration points $n_g$ increases. This fact can be clearly observed when comparing the computation times for settings defined by $n_g=20$ and $n_g=8$. This can be easily explained by the impact of this choice on the  number of required iterations before reaching a solution (see Figures \ref{fig:Niterations}-\ref{fig:Niterations_zoom}).
\e Indeed Figure \ref{fig:Niterations} and its zoomed version depicted on Figure \ref{fig:Niterations_zoom}, show statistics of the number of iterations (necessarily limited to 200 by choice) needed by the different solvers. The curves enable to understand why, in the case of \textbf{SaA}, despite the higher cost of a single iteration induced by the choice $n_g=20$ compared to $n_g=3$, the number of needed iteration is lower as the rate of decrease of the cost function per iteration is higher.
\e Figure \ref{fig:Niterations} also shows that because of the comparatively slow convergence of the \textbf{FGM}, all the settings of this algorithm needed the maximum number of iteration (200) as the stopping criterion is never satisfied for these alternatives. This is also the case for the \textbf{FISTA} when the conservative $L_0=10^6$ is  used. On the other hand, taking too small a value, like $L_0=10^{-6}$ induces higher number of iterations as the backtracking algorithm needs more time to reach a good estimation of the Lipchitz constant $L$. 

\begin{figure}[h!]
    \centering
    \includegraphics[width=\linewidth]{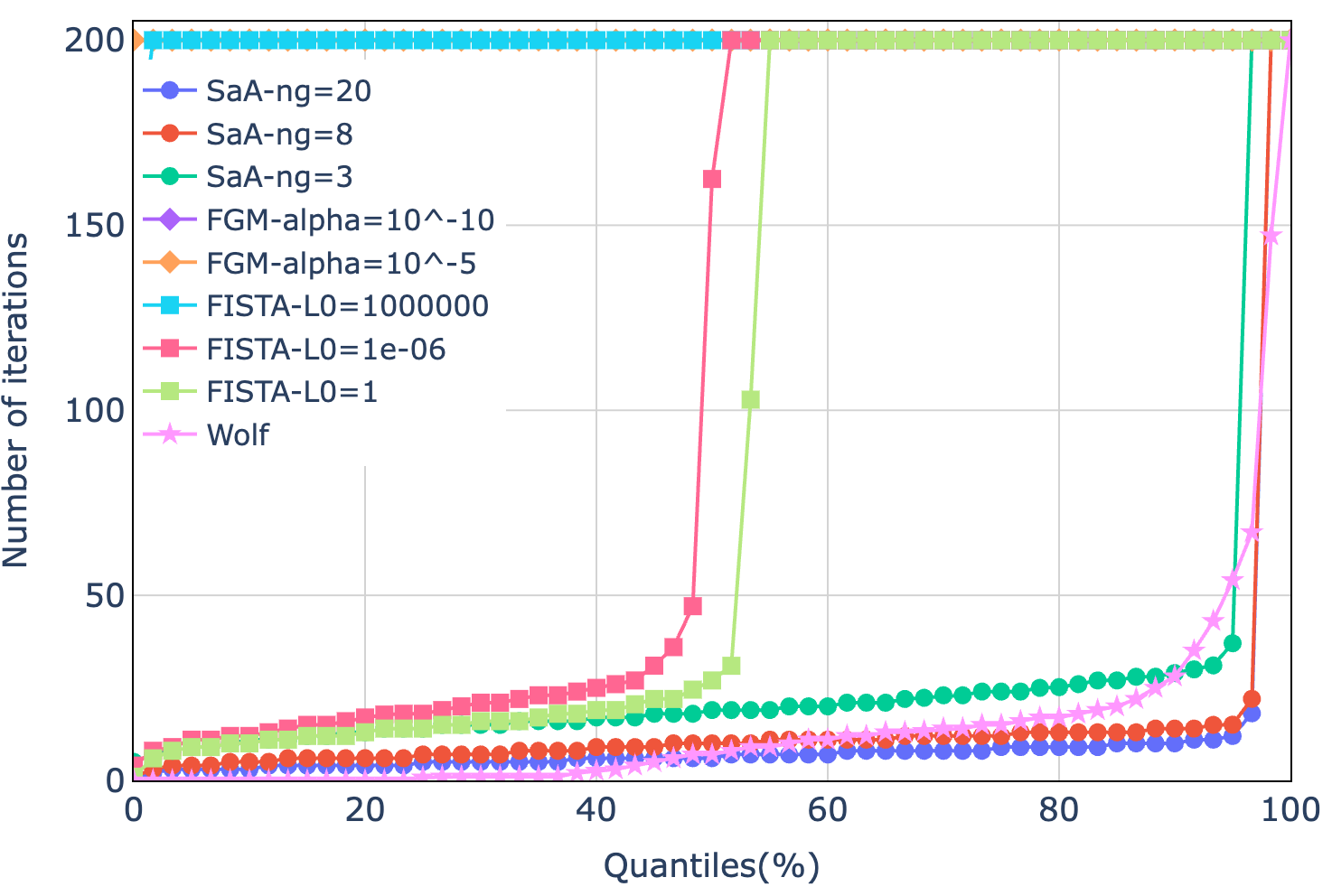}
    \caption{Statistics of the number of iterations (necessarily limited to 200 by choice) needed by the different solvers}
    \label{fig:Niterations}
\end{figure}

\begin{figure}[h!]
    \centering
    \includegraphics[width=\linewidth]{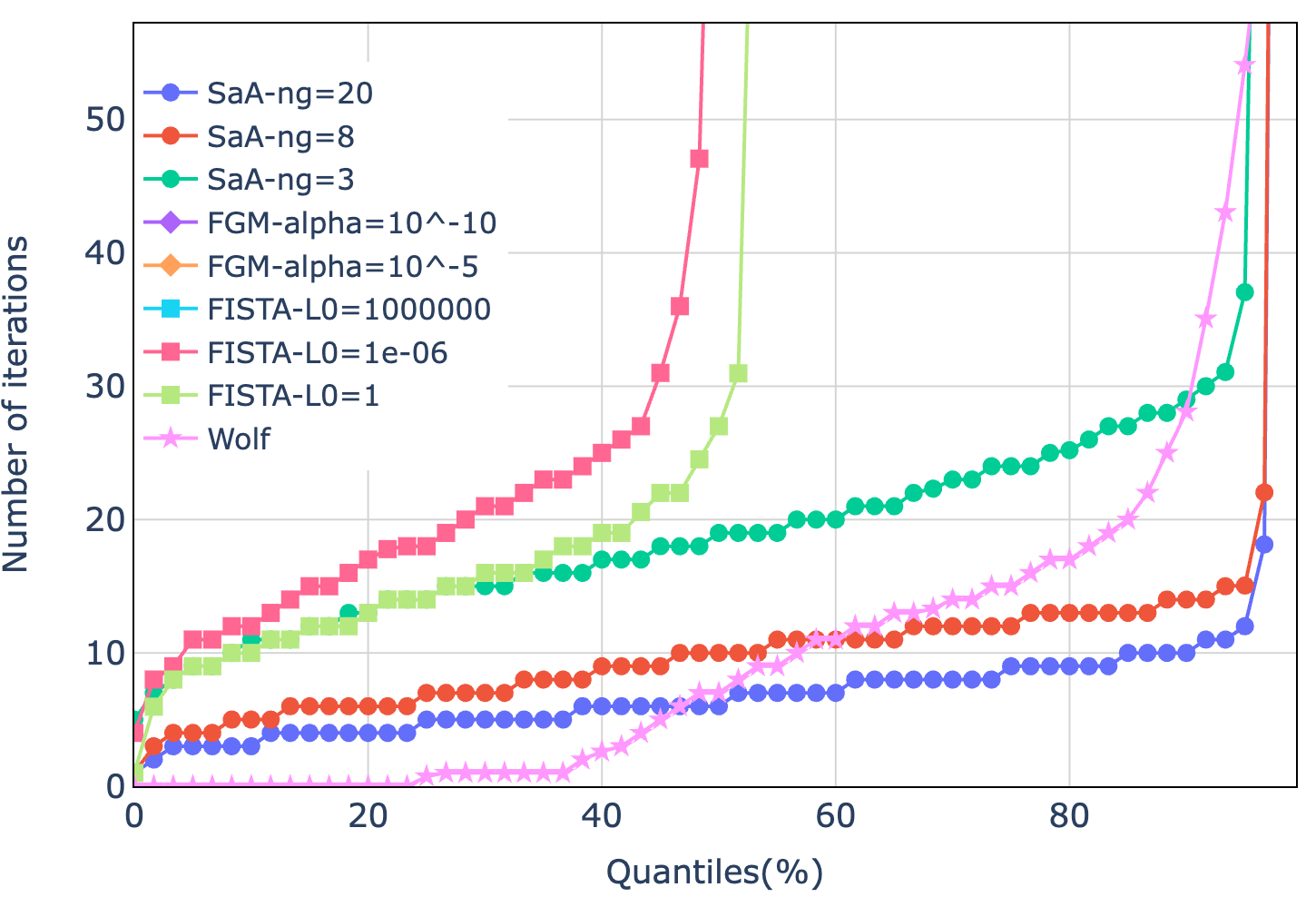}
    \caption{Zoom on the plot of Figure \ref{fig:Niterations} centered on the small number of iterations required by the proposed \textbf{SaA} algorithm.}
    \label{fig:Niterations_zoom}
\end{figure}
\subsubsection{Statistics of the steps along the acceleration path}\label{negative_c}
\e Let us end this section by investigating the statistics of the value taken by the acceleration step, namely the step along the direction $v_{k+1}-v_k$ implemented in step 18: of the algorithm which implements the updating rule \eqref{accelerationstep}. Recall that possible values of this step is controlled by the set of steps included in the set $\mathbb A_\text{lin}$ defined by \eqref{Alin}. For a step equal to $0$ no acceleration is applied and the update resulting from the searched gradient step is adopted. When positive values are used, this implement a positive momentum along the direction $v_{k+1}-v_k$. Finally, negative steps lead to iterate that strictly lie inside the segment $(v_k,v_{k+1})$.
\e The statistics of the 20,400 step values selected by the \textbf{SaA} algorithm over the 600 problems of the benchmark by the three investigate settings (corresponding to $n_g\in \{3, 8, 20\}$ which changes the values in $\mathbb A_\text{lin}$) lead to the normalized histogram shown in Figure \ref{fig:statisticsofc}.
\begin{figure}[h!]
    \centering
    \includegraphics[width=\linewidth]{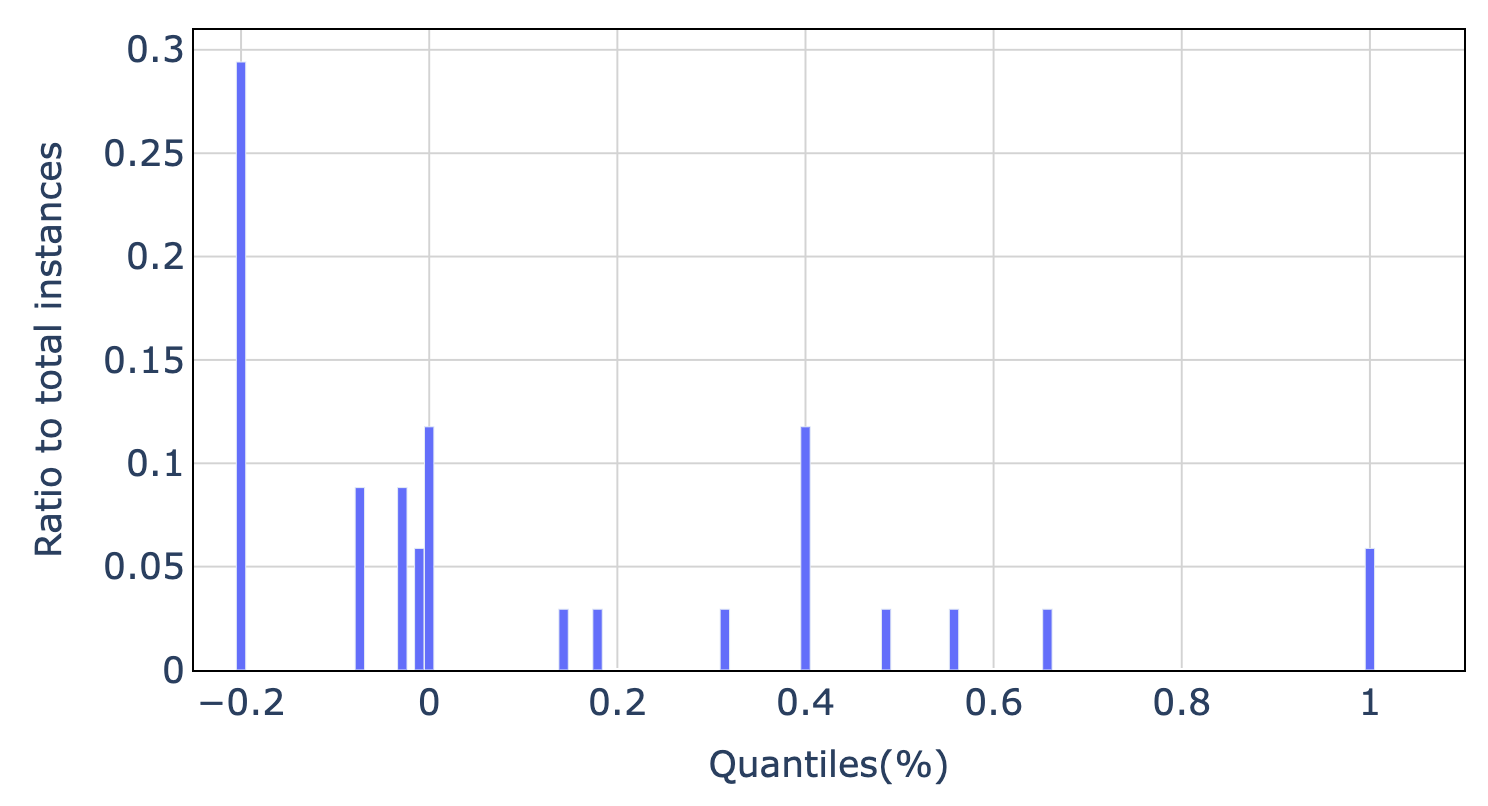}
    \caption{Statistics (Histogram) of the step sizes $c$ used in the acceleration step and taken by the three different versions of the \textbf{SaA} algorithm during the 600 optimization problems contained in the benchmark and for different setting of $n_g\in \{3,8,20\}$ used in the proposed solver.}
    \label{fig:statisticsofc}
\end{figure}
\e The remarkable fact is that the negative value $c=-0.2$ appears to be the most frequent value selected by the algorithm (almost 30\% of the instances). This results is not a priori expected and legitimates the definition of the the set of possible steps $\mathbb A_\text{lin}$ in that it should contain negative values. The presence of some $12\%$ of instances where $c=0$ is chosen corresponds to situations where using the acceleration would have been detrimental to the objective of decreasing the cost function leading to the use of the sole gradient descent step.
\e The examination of the results enables to state the following claim: 
\begin{center}
\begin{tikzpicture}
\node[rounded corners, inner sep=5mm, fill=Black!5, draw=Black] (O){
\begin{minipage}{0.4\textwidth}
Despite the variety of problems included in the benchmark, the default parameters together with any of the suggested values of $n_g$ enable to successfully solve all the problems in less than $200$ iterations. This is obtained thanks to the adaptation mechanisms and the incorporation of cost function exploration at each iteration of the algorithm. 
\e 
This justifies the \textit{parameters-free} qualification of the proposed algorithm that is used in the paper's title. This does not mean the absence of parameters but the robustness of the results to their specific choice making the default values always successful.
\end{minipage}
};    
\node[rounded corners, fill=white, draw=Gray] at(O.north){
\footnotesize \sc  Parameter-free?
};
\end{tikzpicture}
\end{center}
\subsection{Example of use in NMPC real-time implementation}
\noindent In this section, we investigate the use of the proposed gradient based algorithm in implementing constrained NMPC feedback control. But before we dig into this task, let us highlight the following disclaimer: 

\begin{center}
\begin{tikzpicture}
\node[rounded corners, inner sep=5mm, fill=Black!5, draw=Black] (O){
\begin{minipage}{0.4\textwidth}
The message of this section is not to promote the unconditional use of the proposed gradient-based solver in the implementation of NMPC feedback. Rather, it is to consider it as a viable option for \textit{certification} when very short control updating periods are mandatory or simply when standard \textit{advanced} solvers do not meet the real-time implementation requirements.
\end{minipage}
};    
\node[rounded corners, fill=white, draw=Gray] at(O.north){
\footnotesize \sc  Disclaimer
};
\end{tikzpicture}
\end{center}
Notice that the need for certification holds for any NMPC implementation in which the optimization process is interrupted and the necessarily sub-optimal solution is used to define the control to be applied. This is because in such circumstances, the theoretical foundation of the stability are not satisfied regardless of the NMPC formulation being used. For more details about the certification of NMPC, see \cite{ALAMIR201565, alamir-dashboard2025}.
\subsubsection{The NMPC problem}
\noindent For the sake of illustration, we consider the Planar Vertical Take-of and Landing (\textbf{PVTOL}) aircraft example \cite{HAUSER1992665} whose dynamics can be described by the following set of ODEs:
\begin{subequations}
\begin{align}
\ddot y&=-u_1\sin\theta+\epsilon u_2\cos\theta,\label{pvtol1}\\
\ddot z&=+u_1\cos\theta+\epsilon u_2\sin\theta-1,\label{pvtol2}\\
\ddot\theta &= u_2
\end{align}
\end{subequations}
that we need to regulate around $(y,z,\theta)=0$ using bounded control $$u\in [-1.5,+1.5]\times[-0.5, 0.5]$$
while meeting the following constraints on the state (limited rate of change on $z$ and $\theta$):
\begin{equation}
\vert \dot z\vert \le 0.5\quad;\quad \vert \dot \theta\vert \le 0.4\label{NMPCconstraints}
\end{equation}
The NMPC formulation used hereafter is the one recently proposed in \cite{Alamir-relaxing-2025} that shows the advantage of providing stability for prediction horizons that do not meet the reachability condition. The cost function is given by:
\begin{equation}
J(\mathbf u\vert x_0):=\gamma^2\|\dot x_N\|+\gamma \ell(x_N)+\|x_N\|_{Q_f}+\sum_{k=0}^{N}\ell(x_k,u_k) \label{cost}
\end{equation}
where $\ell(x,u)$ incorporates the original regulation cost augmented with the exact penalty on the constraint violation:
\begin{equation}
\ell(x,u):=\|x-x_d\|^2_Q+\|u-u_d\|^2_R +\rho_\text{soft}\times \sum_{i=1}^4\max\bigl\{0,c_i(x)\bigr\}^2\label{defdeell}
\end{equation}
where $c_i(x), i=1,\dots,4$ correspond to the four constraints defined by \eqref{NMPCconstraints}. The steady state is defined by $x_d=(y_d,z_d,0,0,0,0)$ and the steady control is $u_d=(1,0)$.
\e 
In all the numerical investigation $\gamma=200$ and $\rho_\text{soft}=10^7$ are used. The parameters $\epsilon=0.4$ is used in \eqref{pvtol1} and \eqref{pvtol2}. The sampling period $\tau=0.1$ and the  prediction horizon $N=50$ are adopted. 
\subsubsection{The alternative solver: \textbf{fatrop}}
\noindent In order to compare the performance and the computation times of the proposed \textbf{SaA} algorithm with a standard state of the art higher order solvers, the recently proposed \textbf{fatrop}\footnote{This is an Interior-Point algorithm that is dedicated to solving optimal control problem unlike the famous \textbf{ipopt} \cite{wachter2006ipopt} commonly used as the default option within \texttt{Casadi} that is a general purpose interior point algorithm.} algorithm is used in the comparison. The reason for this choice is twofold: 
\begin{itemize}
\item First, this solver is smoothly integrated in the \texttt{Casadi} framework which is also used in the computation of the cost function and the gradient as used by the \textbf{SaA} algorithm. The use of the famous state of the art \textbf{acados} \cite{Verschueren2021acados} solver requires a reformulation of the problem to fit its own syntax. 
\item More importantly, it has been shown that \textbf{Fatrop} provides performance and computation times that are quite comparable to the \textbf{acados} solver. In particular Table IV in \cite{VanroyeLander2023FAFC} where some examples of the same \textit{category}\footnote{number of state, control inputs, prediction horizon, etc.} of problems as the one considered here are investigated, shows that \textbf{fatrop} and \textbf{acados} \cite{Verschueren2021acados} produce quasi identical results within comparable computation times. 
\end{itemize}
When comparing \textbf{SaA} to \textbf{fatrop}, two versions will be compared, namely, the pure python version and the C++/compiled version of each algorithm. More precisely:
\begin{itemize}
\item In the pure python version, the optimization problem is simply formulated using the \texttt{Casadi} framework's \texttt{nlpsol}  without compilation option for \textbf{fatrop} while the cost function and the gradient map are also provided using the \texttt{Casadi}'s \texttt{Function} utility (without compilation option) starting from the same definitions of the maps used as argument to \texttt{nlpsol} to set the \textbf{fatrop} solver. 
\item In the compiled versions, the \texttt{nlpsol} is called with the following options dictionary:
\begin{center}
\begin{tikzpicture}
\node[rounded corners, fill=Black!5, draw, inner sep=6mm](T){
\begin{minipage}{0.3\textwidth}
\footnotesize
\begin{lstlisting}[language=Python]
jit_opts = {
            "jit": True,
            "compiler": "shell",
            "jit_options": {
                "flags": '-O2'
                }
            }
\end{lstlisting}
\end{minipage}
};
\node[rounded corners, fill=white, draw=Black]at(T.north){
\footnotesize Compilation option
};
\end{tikzpicture}
\end{center}
which provides an optimized \texttt{casadi} native compiled version of the \textbf{fatrop} algorithm. As for the \textbf{SaA} algorithm, the same compilation options is used in the creation of the cost function and the gradient through the \texttt{Function} utilities provided by \texttt{Casadi}. This is not necessarily the most optimized version as all the remaining computations involved in the \textbf{SaA} such as the multiple calls for the cost function during the two line search operation are still coded in python leading to an extra calls-related overhead that can be removed by totally compiling the whole algorithm. 
\end{itemize}
\begin{rem}[Why to show pure python comparison also]
It is important to underline that while the compilation of the \textbf{SaA} algorithm needs few seconds, the one involved in the use of the above mentioned compilation option in \texttt{nlpsol} requires around 30 min of processing\footnote{On a MacBook Pro, Apple M3 Pro, 18 Go under Tahoe 26.5.1}. This can become cumbersome during the mandatory prototyping phase. That is the reason for which comparisons are also given for pure python versions of the algorithms. 
\end{rem}
\subsubsection{Protocol and metrics for the comparison}
\noindent Without loss of generality, the targeted steady state is taken equal to $x_d=0\in \mathbb R^6$ and $u_d=(1,0)$ and only the initial state is selected to define a scenario. This can recover any possibe pair of $(x_0,x_d)$ of initial state and desired state by simple translation. 
\e A number  $\texttt{nSamples}=100$ of scenarios are generated corresponding to different initial states inside $[-0.5,0.5]^6$. For each scenario, closed-loop trajectories under the different NMPC controllers are simulated over $N_{sim}=251$ sampling periods and the resulting closed-loop trajectories and computation times are compared following the protocol defined below. 
\e 
First of all, in order to compare the convergence behavior, a convergence-related scalar is associated to each closed-loop trajectories using the following steps:
\begin{itemize}
\item The sequence of values of the \textit{returned} sub-optimal open-loop cost obtained by the optimizer at each instant along the closed-loop trajectory is obtained, say $J^\text{cl}_k$. For well formulated NMPC with completed optimization, this sequence should be strictly decreasing.
\item A normalized sequence is obtained by using:
\begin{equation}
\bar J_k^{cl}:=\dfrac{J_k^\text{cl}}{J_0^\text{cl}} \label{defdeJkbar}
\end{equation}
For typically successful closed-loop, this sequence is strictly decreasing with $1$ as initial value. 
\item One way to associate a scalar that summarizes the possible contraction along the sequence over the whole closed-loop trajectory is to find the maximum rate of exponential decrease that upper bounds the normalized sequence, namely:
\begin{equation}
\mu^*= \max\Bigl\{\ \mu\quad \text{s.t.}\quad \forall k, \bar J_k^\text{cl}\le e^{-\mu k}\ \Bigr\}\label{defdemstar}
\end{equation}
Obviously, the larger (positive) $\mu^*$ the faster is the convergence over the duration of the closed-loop simulation scenario.
\item The comparison of the performances of two algorithms, say $A_1$ and $A_2$ is obtained by observing the statistics (percentiles) of the ratios $\mathcal R_{A_1/A_2}$ between their $\mu^\star$-associated values over the 100 scenarios used in the investigation, namely the percentiles of the quantity:
\begin{equation}
\mathcal R_{A_1/A_2} :=  \dfrac{\mu^\star_{A_1}}{\mu^\star_{A_2}}\label{defdeeta}
\end{equation}
\end{itemize}
While the percentiles of the $\mathcal R_{A_1/A_2}$ indicator enable to compare the convergence rates, the key advantage of using the gradient-based algorithm is related to the minimal achievable computation time. This will be compared by examining the histogram of the computation times needed at each updating periods over the closed-loop simulations of all the 100 scenarios. 
\e 
But before digging into such \textit{blind} statistics over many scenarios, let us give typical instances of closed-loop trajectories during a single closed-loop scenario. This is depicted in Figure \ref{fig:typivalclosedloopsc} which shows a single scenario under compiled multiple-shooting \textbf{fatrop} algorithm (with \texttt{maxIter}=1) on one hand and under the proposed compiled \textbf{SaA} algorithm (with \texttt{maxIter}=5)  on the other hand. 
\begin{figure}[h!]
    \centering
    \footnotesize (Multiple shooting compiled \textbf{fatrop} | \texttt{maxIter}=1)\\
    \includegraphics[width=0.99\linewidth]{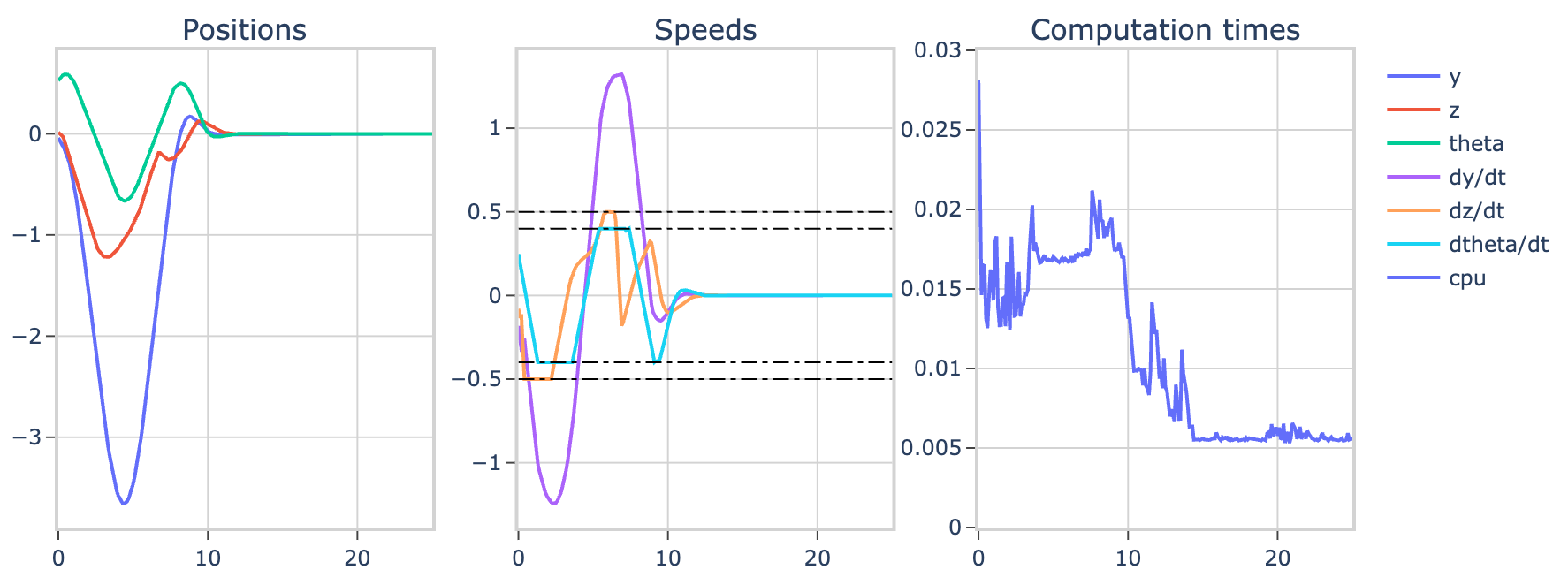}
    \footnotesize (Proposed gradient based \textbf{SaA} | $n_g=8$ | \texttt{maxIter}=5)\\
    \includegraphics[width=0.99\linewidth]{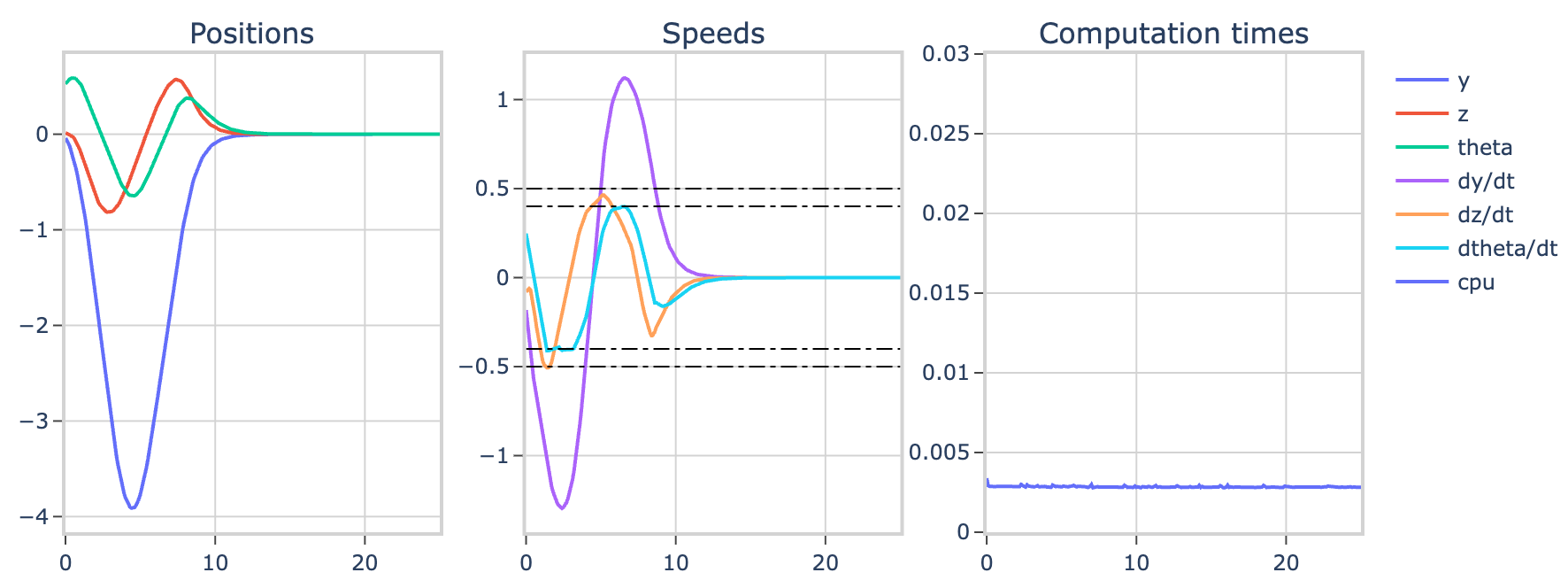}
    \caption{Typical comparison between closed-loop trajectories under \textbf{compiled} multiple shooting \textbf{fatrop} and a compiled version of the proposed gradient based \textbf{SaA} algorithm. Notice the soft constraints satisfaction, the slight deterioration of performance against drastic (8-9 times faster) drop in the worst computation time. Dash-dotted lines without legend represent the constraints on the state.}
    \label{fig:typivalclosedloopsc}
\end{figure}
\e This typical example show a commonly encountered behavior where at the cost of possible slight decrease of performance, the gradient-based algorithm offers a drastic reduction of the elementary computation time. Notice that the use of \texttt{maxIter}=1 in the case of \textbf{fatrop} shows the incompressible amount of computation needed for a single iteration of this algorithm.
\begin{figure}[h!]
    \centering
    \footnotesize (Multiple shooting compiled \textbf{fatrop} | \texttt{maxIter}=1)\\
    \includegraphics[width=0.99\linewidth]{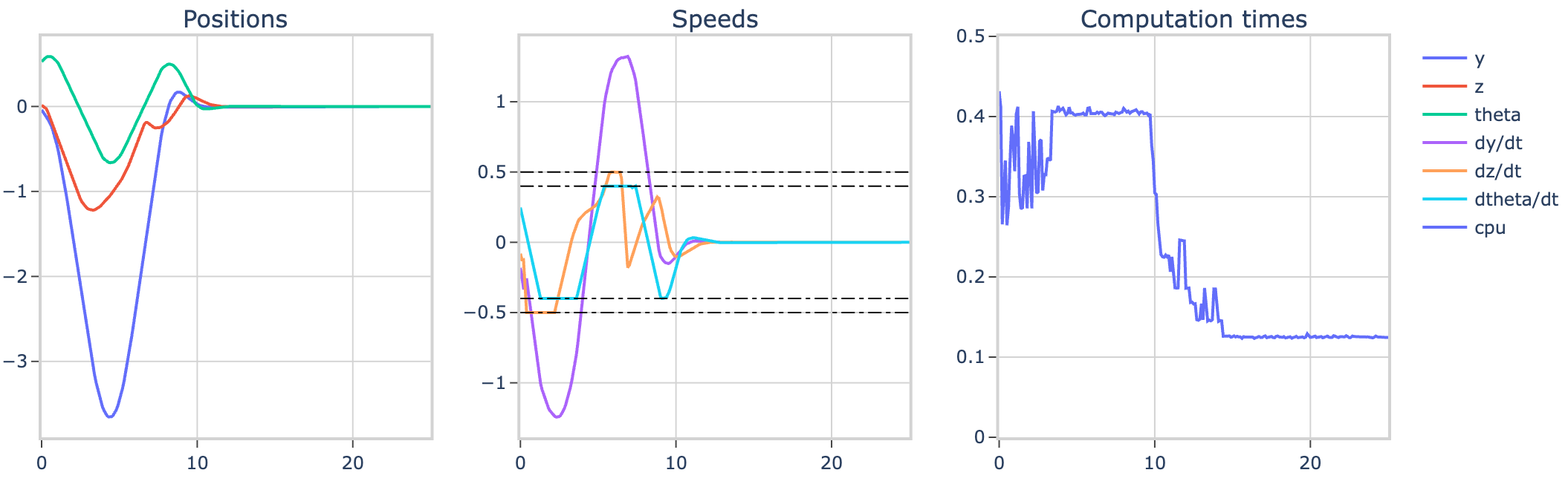}
    \footnotesize (Proposed gradient based \textbf{SaA} | $n_g=8$ | \texttt{maxIter}=10)\\
    \includegraphics[width=0.99\linewidth]{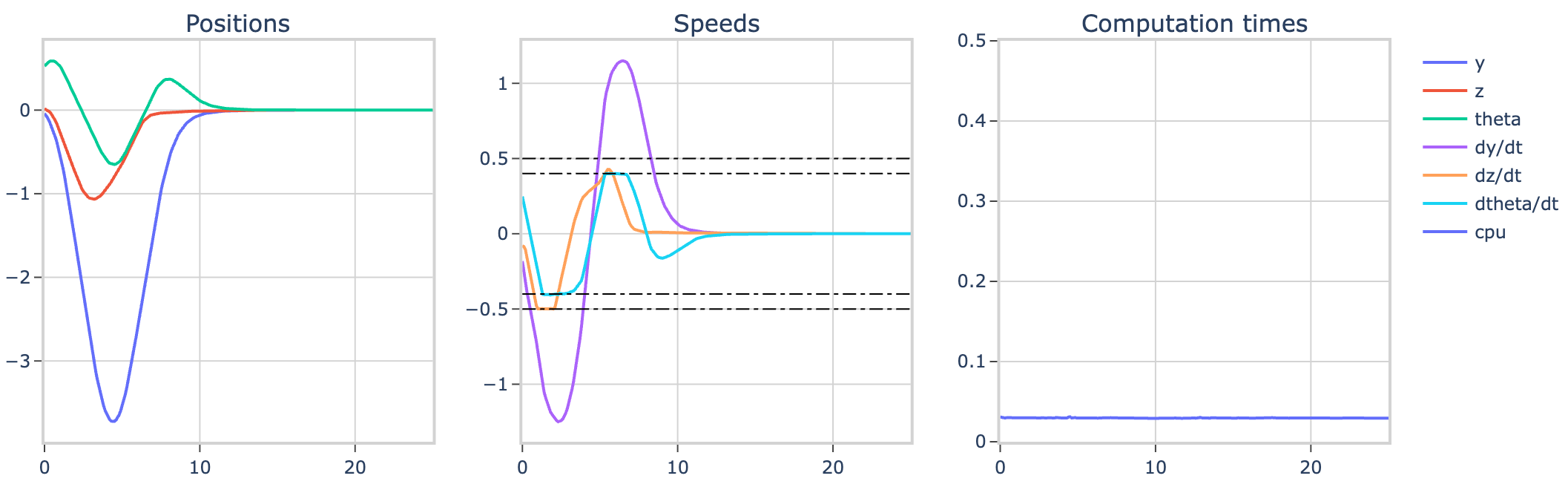}
    \caption{Typical comparison between closed-loop trajectories under \textbf{pure-python} multiple shooting \textbf{fatrop} and a pure python version of the proposed gradient based \textbf{SaA} algorithm. Notice the soft constraints satisfaction, the slight deterioration of performance against drastic  drop in the worst computation time. Dash-dotted lines without legend in the speeds-related subplots represent the constraints on the state.}
    \label{fig:typivalclosedloopscpp}
\end{figure}
The same comparison is reported on Figure \ref{fig:typivalclosedloopscpp} in the case where non compiled pure-python versions of the algorithms are used. Notice that the use of \texttt{maxIter}=10 in \textbf{SaA} (instead of \texttt{maxIter}=5 in Figure \ref{fig:typivalclosedloopsc}) leads to closer trajectories to the ones obtained under \textbf{fatrop} and to tighter fulfillment of the state constraints while keeping a drastic increase in the computation time that can be quite welcome in the prototyping phase avoiding the rather long compilation time. 
\e 
Figure \ref{fig:ratios_1} shows the percentiles of the contraction's exponent ratio $\mathcal R_{\textbf{SaA}/\textbf{fatrop}}$ defined by \eqref{defdeeta}, namely:
\begin{equation}
\mathcal R_{\textbf{SaA}/\textbf{fatrop}} = \frac{\mu^\star_{\textbf{SaA}\vert \texttt{maxIter}=3,5,10}}{\mu^\star_{\textbf{fatrop}\vert \texttt{maxIter=1}}} \label{defdeRplote}
\end{equation}
for different values o the \texttt{maxIter} arguments for \textbf{SaA} while  \texttt{maxIter}=1 is systematically used for \textbf{fatrop} in order to seek minimum computation time.
\begin{figure}[h!]
    \centering
    \includegraphics[width=0.6\linewidth]{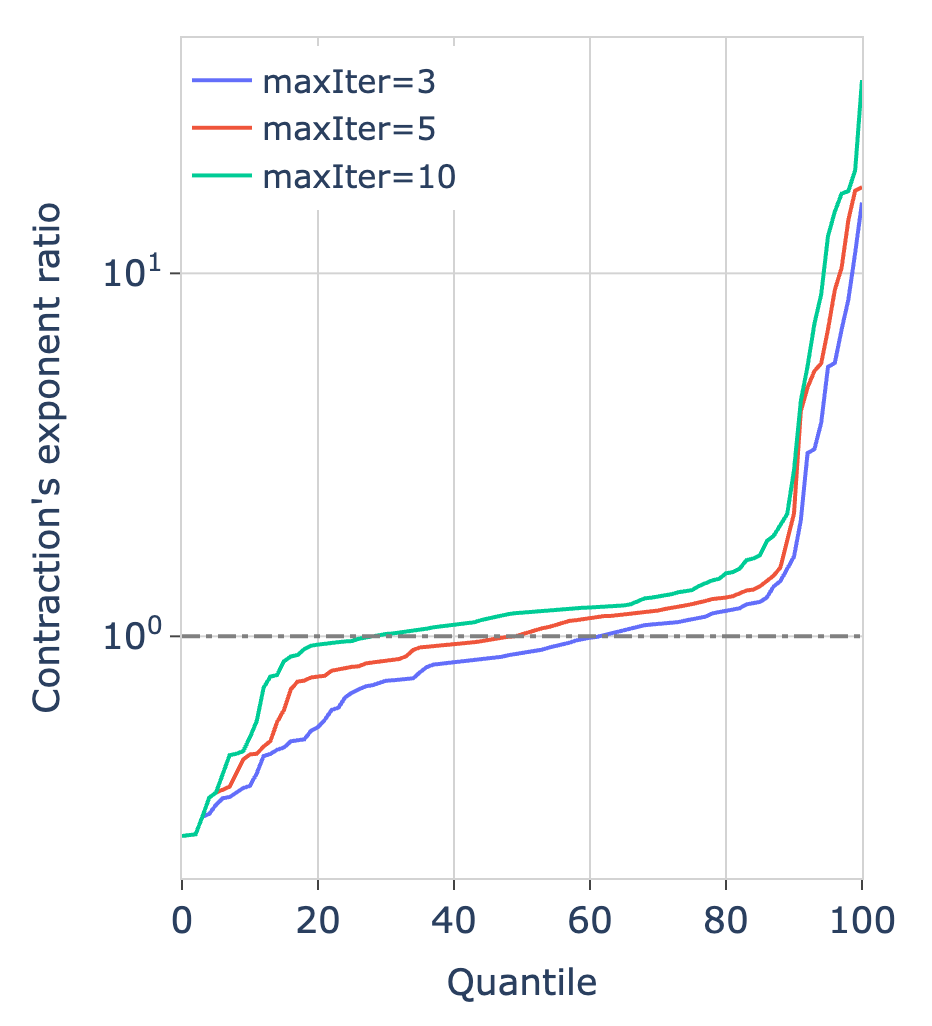}
    \caption{Percentiles of the contraction's exponent ratio $\mathcal R_{\textbf{SaA}/\textbf{fatrop}}$ defined by \eqref{defdeeta} for different values of the \texttt{maxIter} argument for \textbf{SaA} while \texttt{maxIter}=1 for \textbf{fatrop}.}
    \label{fig:ratios_1}
\end{figure}
\e 
Recall that the higher the ratio the better the contraction compared to the reference method (here \textbf{fatrop} with \texttt{maxIter}=1). On the other hand, Figure \ref{fig:ratios_1} shows the expected fact that when more iterations are allowed, the contraction per sampling period is better. 
\e 
It is noticeable to observe that for the case where \texttt{maxIter}=5 is used for \textbf{SaA}, almost exactly 50\% of the scenarios shows \textit{higher contraction} of \textbf{SaA} while the remaining 50\% gives advantages to the \textbf{fatrop} algorithm. These percentage shift to (70\%, 30\%) and (40\%, 60\%) for \texttt{maxIter}=10 and \texttt{maxIter}=3 respectively.
\e 
Now since Figure \ref{fig:ratios_1} shows only a ratio, it is interesting to show the percentile of the values taken by the denominator in \eqref{defdeRplote}. This is shown in Figure \ref{fig:mustar_ref} where the percentiles of the contraction exponent of the \textbf{fatrop} algorithm, namely:
$$\mu^\star_{\textbf{fatrop}\vert \texttt{maxIter=1}}$$ involved in \eqref{defdeRplote} are shown. 
\begin{figure}[h!]
    \centering
    \includegraphics[width=0.6\linewidth]{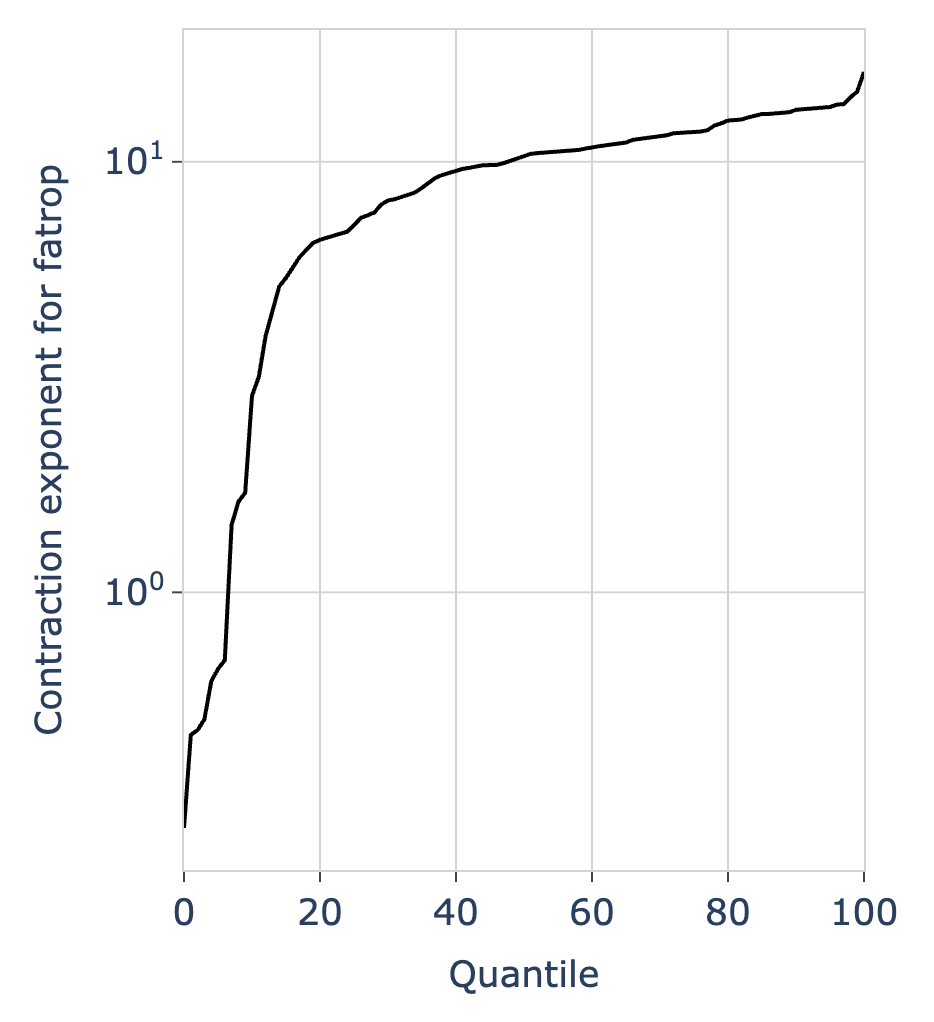}
    \caption{statistics (percentiles) of the contraction rate $\mu^\star_{\textbf{fatrop}\vert \texttt{maxIter=1}}$ with respect to which the ratios shown in Figure \ref{fig:ratios_1} are computed. The minimal value ($\approx 0.28$) indicates that contraction arises in all the 100 scenarios used in the investigation.}
    \label{fig:mustar_ref}
\end{figure}
\e 
This Figure shows in particular that the minimum value of the contraction's  exponent over all the scenario is greater than 0.28. This means that in all the scenarios, contraction is effective and this holds for all the algorithms referenced in Figure \ref{fig:ratios_1}.
\begin{rem}[soft constraints $\rightarrow$ different cost function]
It is important to notice however that because of the use of soft constraints in \textbf{SaA}, the cost function is not necessarily the same  for both algorithms should the initial guess leads to constraints violation. Nevertheless, as the value of each cost function is normalized by its own initial value, the contraction concept is relevant for each algorithm. Moreover, as the inequality used in the definition of the bound in \eqref{defdemstar} applies to the whole rather long scenario at the end of which the constraints are generally inactive, the two cost functions becomes rigorously identical over a large part of the closed-loop simulation. 
\end{rem}
\begin{rem}[The systematic use of \texttt{maxIter}=1 for \textbf{fatrop}] It is important to keep in mind that our aim here is to compare the extent to which the elementary computation time of the reference solver can be reduced. More importantly, the aim is to see to which extent, the \textbf{SaA} is able to achieve successful closed-loop while requiring much smaller elementary computation times. Hence by using \texttt{maxIter}=1 for \textbf{fatrop} enables to show that for this solver, the minimum amount of elementary computation time reaches a kind of \textit{hard} lower bound that cannot be reduced. It is only then that the relevance of the gradient-based algorithm \textbf{SaA} becomes obvious. In other words, no matter how much increase in the performance of \textbf{fatrop} can be reached by increasing its \texttt{maxIter} parameters, the corresponding computation time can only be higher and hence not eligible in some circumstances.
\end{rem}
Now that we can shortly state that in terms of the performance, (\textbf{fatrop} | \texttt{maxIter}=1) and (\textbf{SaA} | \texttt{maxIter}=5) are quite comparable, it is time to examine the statistics of the computation times. 

\begin{figure}[h!]
\begin{center}
\includegraphics[width=0.35\textwidth]{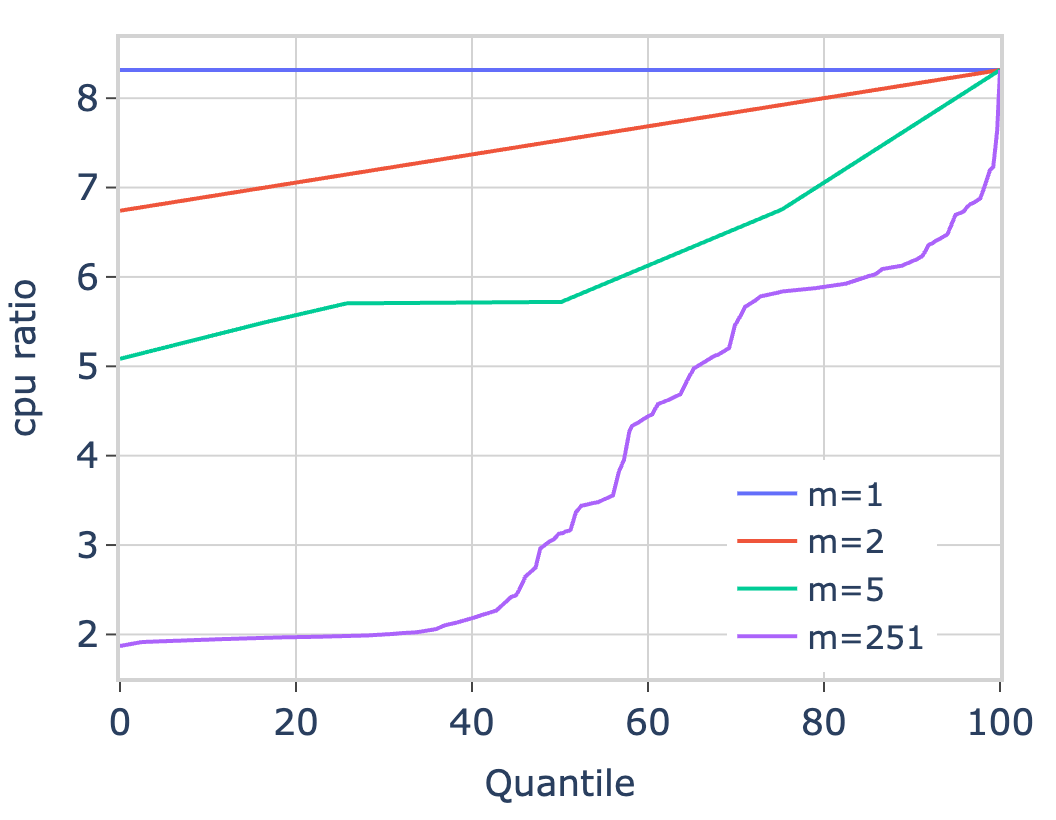} 
\end{center}
\caption{Statistics (percentiles) of the ratios of computation times ($\texttt{cpu}_\textbf{fatrop} / \texttt{cpu}_\textbf{SaA}$) over all the first $m$ sampling periods of all the 100 scenarios for the \textbf{compiled} versions of the algorithms.}\label{fig:ratio_cpu}
\end{figure}
\e Figure \ref{fig:ratio_cpu} shows the percentiles of the ratios of computation times of the (\textbf{fatrop} | \texttt{maxIter}=1) algorithm relative to those of the (\textbf{SaA} | \texttt{maxIter}=5) algorithm for all the first $m$ sampling periods of all the 100 simulated closed-loop scenario and this for different values of $m\in \{1,2,5,251\}$. 
\e Notice that for $m=251$, the percentiles are computed for the total number of sampling periods. While for the smaller values of $m\in \{1,2,5\}$, the statistics incorporate the very first sampling periods of each scenario where the computation times take systemically the highest values. 
\e The plots of Figure \ref{fig:ratio_cpu} suggests that when the computation times are high (during the few sampling instants after a change in the set-point), the reduction in the computation time offered by the \textbf{SaA} gradient-based algorithm, compared to the \textbf{fatrop} algorithm, reach values ranging from a factor of 5 to 9. For stationary regime periods, the reduction factor is around 2, as it can be clearly seen on the typical example shown in Figure \ref{fig:typivalclosedloopsc}.
\e Finally, Figure \ref{fig:ratio_cpu_python} shows the same results as the ones shown in Figure \ref{fig:ratio_cpu} in the case where non compiled pure python versions of the algorithms are used. One can observe a higher rate of reduction in the computation times suggesting that the compilation of the gradient-based algorithm is not optimized as the one used in the \texttt{Casadi} utility. 
\begin{figure}[h!]
\begin{center}
\includegraphics[width=0.35\textwidth]{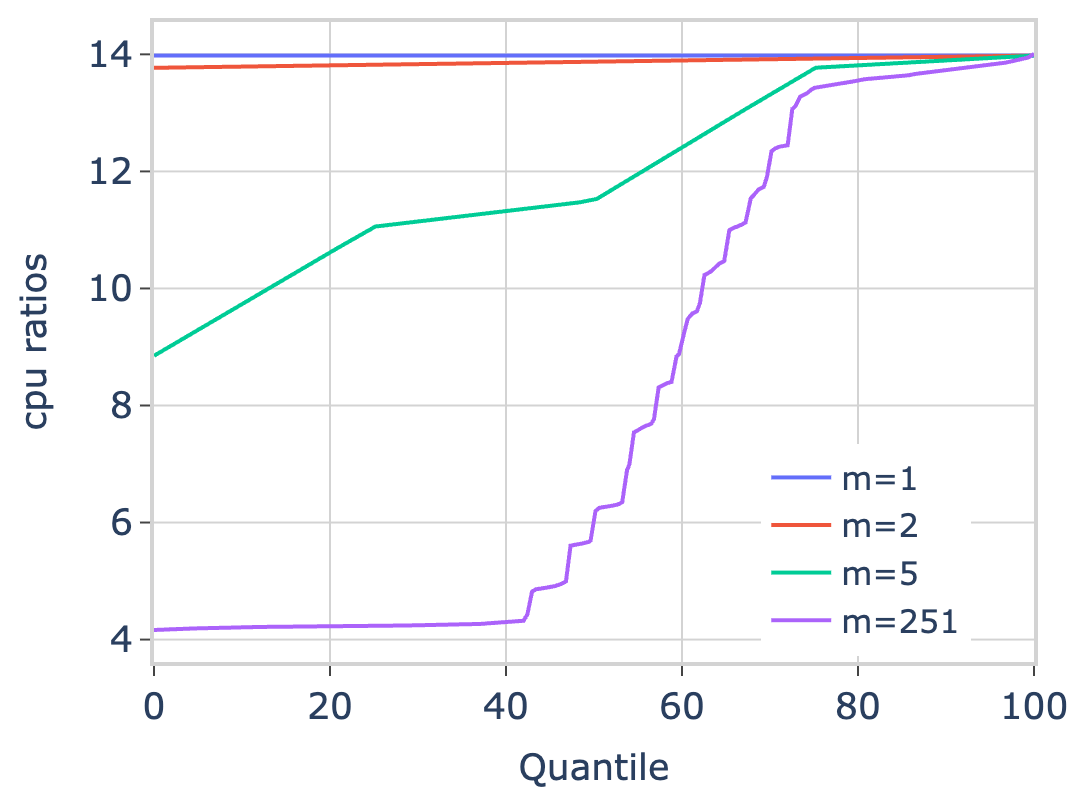} 
\end{center}
\caption{Statistics (percentiles) of the ratios of computation times ($\texttt{cpu}_\textbf{fatrop} / \texttt{cpu}_\textbf{SaA}$) over all the first $m$ sampling periods of all the 100 scenarios for the \textbf{pure python} versions of the algorithms.}\label{fig:ratio_cpu_python}
\end{figure}
\section{Conclusion and future work}\label{sec-conc}
\noindent In this paper, a new gradient-based algorithm is proposed for box-constrained optimization problems for which the cost function and its gradient are available. 
\e Unlike existing gradient-based alternatives, the proposed algorithm, termed \textbf{SaA} for Search \& Accelerate combines a novel line search (with a trust region adaptation mechanism) along the gradient and another novel line search over the Nesterov acceleration path while being supported by a rigorous convergence proof.
\e A benchmark including 600 box constrained optimization problems is also proposed and made publicly available showing that the proposed algorithm outperforms existing alternatives in both computation time and solution quality. 
\e Finally, the relevance of the proposed gradient-based algorithm \textbf{SaA} for NMPC implementation when very small control updating periods are required has been shown through a dedicated case study.
\e Undergoing work is twofold: first, a sparsity-aware version of the algorithm is under investigation; then the application of the algorithm in efficiently training (dense/sparse) neural networks will be undertaken.

\bibliographystyle{elsarticle-harv}  
\bibliography{bib_derivation}            

\end{document}